\documentclass{article}
\makeatletter
\renewcommand*{\@fnsymbol}[1]{\ensuremath{\ifcase#1\or *\or 1\or 2\or
 3\or 4\or 5\or 6\or 7 \or 8 \else\@ctrerr\fi}}
\makeatother
\usepackage{amsmath,amsthm,amssymb}
\usepackage{upgreek}
\usepackage{cite}
\usepackage{enumerate}
\newtheorem{theorem}{Theorem}[section]

\newtheorem{lemma}[theorem]{Lemma}
\newtheorem{proposition}[theorem]{Proposition}
\theoremstyle{definition}

\theoremstyle{remark}
\newtheorem{remark}[theorem]{Remark}

\allowdisplaybreaks[3]
\numberwithin{equation}{section}
\DeclareMathOperator*{\esssup}{ess\,sup}

\newcommand\R{{\mathbb R}}
\newcommand\X{{\R^d}}
\newcommand\N{{\mathbb N}}

\newcommand\F{{\mathcal F}}

\newcommand\B{{\mathcal B}}
\newcommand\K{{\mathcal K}}
\newcommand\Ph{{\mathbb{\phi}}}

\newcommand\La{\Lambda}
\newcommand\la{\lambda}
\newcommand\Ga{\Gamma}
\newcommand\ga{\gamma}
\newcommand\eps{\varepsilon}

\newcommand\Bbs{B_\mathrm{bs}(\Ga_0)}
\newcommand\Fcyl{\F_{\mathrm{cyl}}(\Ga)}

\author{Dmitri Finkelshtein\thanks{Department of Mathematics,
Swansea University, Singleton Park, Swansea SA2 8PP, U.K. ({\tt d.l.finkelshtein@swansea.ac.uk}).} \and Yuri Kondratiev\thanks{Fakult\"{a}t
f\"{u}r Mathematik, Universit\"{a}t Bielefeld, 33615 Bielefeld,
Germany ({\tt kondrat@math.uni-bielefeld.de}).} \and Oleksandr
Kutoviy\thanks{Department of Mathematics, Massachusetts Institute of Technology,
77 Massachusetts Avenue E18-420, Cambridge, MA, USA ({\tt
kutovyi@mit.edu}); Fakult\"{a}t f\"{u}r Mathematik, Universit\"{a}t
Bielefeld, 33615 Bielefeld, Germany ({\tt
kutoviy@math.uni-bielefeld.de}).} \and Elena
Zhizhina\thanks{Institute for Information Transmission Problems,
Moscow, Russia ({\tt ejj@iitp.ru}).}}

\title{On an aggregation in birth-and-death stochastic dynamics\thanks{The financial
support of DFG through the SFB 701 (Bielefeld University), grant RFBR 13-01-12410 OFI-M, and NSF Grant DMS-1008132  is gratefully
acknowledged.}}

\begin{document}

\maketitle

\begin{abstract}
We consider birth and death stochastic dynamics
of particle systems with attractive interaction.
The heuristic generator of the dynamics has
a constant birth rate and density dependent decreasing
death rate.
%that defines prolongation of life for elements
%of configuration provided that there are many other elements nearby.
%We
The corresponding statistical dynamics is constructed. Using the Vlasov type scaling we derive the limiting mesoscopic
evolution and prove that this evolution propagates chaos.
We study a non-linear non-local kinetic equation for
the first correlation function (density of population). The existence of
uniformly bounded solutions as well as solutions
growing inside of a bounded domain and expanding in the space are shown.
These solutions describe two regimes in the mesoscopic system:
regulation and aggregation.
\end{abstract}

\section{Introduction}

The notion of an aggregation has many particular forms depending on the class of problems under considerations. One particular case is related to the study of the animal motion in ecology where the aggregation is a process of finding a higher density of animals at some place comparing to the overall mean density. Usually, this notion is described by a heuristic equation for the density, e.g., of reaction-convention-diffusion type, see e.g. \cite{DdF2008,BCM2007,BB2010,Eft2012} and the literature therein. The problem of a derivation of such equation from underlying interacting particle stochastic dynamics was discussed recently in \cite{BHW2012}. In this paper, the authors considered the case of interacting diffusions for individual based models and the evolution equation for the population density appeared in a mean field limit as a kinetic equation for the microscopic model. The derivation of this equation was done on a physical level of the rigor and the analysis of the resulting evolution equation was realized essentially by numeric simulations. But an aggregation may be observed also in many other stochastic dynamics of interacting particle systems. The aim of this work is to give a concrete illustration of aggregation effect.

We will deal with birth-and-death Markov dynamics in the continuum. The phase space of such processes is the space $\Ga=\Ga(\X)$ of locally finite configurations (subsets) in the Euclidean space $\X$. A structural description of considered processes may be given by means of their heuristic generators which on proper functions (observables) $F:\Ga\to\R$ have the following form
\begin{multline}\label{BADG}
(LF)(\ga)= \sum_{x\in\ga} d(x,\ga) [F(\ga\setminus \{x\})-F(\ga)] \\+\int_{\X} b(x,\ga)[F(\ga\cup \{x\})-F(x)] dx.\end{multline}
Here nonnegative functions $d$ and $b$ are given death and birth rates. We will write informally $L=L_d +L_b$ to separate birth and death parts. Note that at the present time the existence problem for Markov processes in $\Ga$ for a given general birth and death rates is an essentially open problem. We refer to \cite{GK2006} for a detailed discussion of this point. An alternative way of studying the system is to consider rather statistical dynamics then stochastic (i.e., Markov process). The latter means that we are interested in construction and properties of solutions to the corresponding forward Kolmogorov (Fokker--Planck) equation. It describes a time evolution of initial states (initial distributions) of the system, see \cite{KoKutZh,KKM09,FKK2011a} for details.

For particular choices of rates the situation with the construction of corresponding Markov dynamics may be essentially simpler. For example, let $d\equiv 0$ and $b\equiv \uplambda >0$. Then the related Markov process describes independent birth of particles in $\X$ with uniformly distributed birth locations and activity parameter $\uplambda $ (the latter process may be considered also as independent immigration of particles). It is easy to see that the density of particles in such process will grow linearly, see e.g. \cite{FK2009}. A presence of a nontrivial death rate leads to the regulation of the population. Consider the case of constant $d\equiv m>0$ that is interpreted as mortality of existing particles after independent exponentially distributed (with parameter $m$) random moments of time. Then the generator $L_m+L_\uplambda $ gives correctly defined Surgailis process of independent birth and death events in $\Ga$ and this process has a unique invariant measure on $\Ga$ that is just Poisson measure with the constant intensity $\uplambda /m$, see \cite{Fin2010Surg,Sur1983,Sur1984}. Therefore, we observe here a self-regulation of the density of particles due to the death part of the generator.

To create the notion of aggregation, we will consider a population with mutualism that produces lower mortality in the dense part of a population. To~realize this effect in our model, we take a potential $\phi:\X \to \R_+$ given by
an even function $\phi\in L^1 (\X)$. Say, for simplicity, $\phi $ is a continuous function with compact support. For $x\in \ga$ set
$$
E(x, \ga \setminus x):= \sum_{y\in \ga\setminus x} \phi(x-y).
$$
The death rate is defined as
\begin{equation}
\label{DR}
d(x,\ga):= \exp(-E(x,\ga\setminus x)).
\end{equation}
This rate has the following property: for $\gamma, \gamma' \in \Ga$, the inclusion $\gamma\subset \gamma'$ yields
$d(x,\gamma') \leq d(x,\gamma)$, i.e., the mortality of an individual $x$ is decaying with the growth of the population. Such models belong to the class of systems with attraction in the terminology of interacting particle systems theory, see e.g. \cite{Lig1985}.
Note that (contrary to the lattice case) interacting particle systems with the attraction property in the continuum form rather exotic
class. For example, the heuristic invariant measure for our Markov generator $L:=L_d+L_\uplambda $ must be the grand canonical Gibbs measure with the density $\uplambda $
and the interacting potential $-\phi$. But such a measure does not exist! It is why models with attraction in the continuum may
have unusual properties comparing with typical statistical physics systems.

Let $\Lambda=B(x_0,r)\subset \X$ be a ball such that
$$
\inf_{x,y\in \Lambda} \phi(x-y)\geq c_{\Lambda}>0.
$$
Then, for the function $V(\gamma):=|\gamma\cap \Lambda|$ ($|A|$ denotes the cardinality of a set $A$),
we have
$$
(LV)(\ga)\geq - V(\ga) e^{-c_{\Lambda}V(\ga)} +\uplambda m(\Lambda),
$$
where $m(\Lambda)$ denotes the Lebesgue measure of $\Lambda$. Assuming the existence of the Markov process $\ga_t$ for our model and using a reverse Gronwall inequality, we easily deduce that for the large enough
initial number of individuals $V(\ga)$ in $\Lambda$ for the initial configuration $\ga$ the expected number of
individuals in $\Lambda$, that is
$E^{\ga} V(\ga_t)$, will grow with time. It means that an initial large enough fluctuation of particles will grow that is an aggregation effect on the microscopic level. But such fluctuation always will appear in our process with a positive probability.
Unfortunately, this heuristic consideration one can not make rigorous because, first of all, the existence of corresponding Markov process is still an open problem.
Secondly, in such microscopic approach any information about the growth of an initial fluctuation in the space
is difficult to obtain. We shall expect that not only number of particles in a given initial volume but
also the size of the region with the high density will grow with time.

Instead of this microscopic ``king's way'' of the analysis of our system, we will realize the following program:

\begin{enumerate}[(i)]
\item we will construct a statistical dynamics of the system for certain class of initial states;

\item using a Vlasov type scaling for the statistical dynamics we will derive the limiting mesoscopic hierarchy
for the correlation functions; the convergence of rescaled dynamics to the solution to the limiting hierarchy will be proven rigorously;

\item it will be shown that this hierarchy has a chaos preservation property;

\item the latter property produces a kinetic equation for the density of population;

\item we will analyze this kinetic equation in detail showing an aggregation notion.

\end{enumerate}

Note that the considered model admits another interpretation in the framework of mathematical physics. Let us consider a process of the snow rain where the snow particles appear on a surface randomly and uniformly
distributed. The intensity of the snow melt in a point of the surface depends obviously on the density of the snow layer
around this point: higher density will effect slower melt. Then an aggregation may be considered as an appearing of snowbanks on the surface and may be described by means of our model.

The paper consists of two parts. The first part is devoted to the mathematically rigorous realization of the items (i)--(iv) above. Namely, in Section~2, we briefly describe the background for analysis on the configuration space $\Ga$ (more detailed explanation can be found in e.g. \cite{KK2002,KoKut}). In Section~3, we construct the microscopic dynamics in the sense that we solve the evolution equation for correlation functions of our systems. This equation is an analog of the well-known BBGKY-hierarchy for Hamiltonian dynamics. The solution exists in a space of correlation functions with the so-called Ruelle bounds but on a finite time interval only. In Section 4, we study a mesoscopic description for our systems, which is based on the approach proposed in \cite{FKK2010a} and realized for some particular models in e.g. \cite{FKK2011a,FKK2010b,FKKoz2011}. As a result, we obtain a reduced (limiting) dynamics of correlation functions, however, this dynamics has the so-called chaotic preservation property. The latter is that the Poisson (free) distributions are preserved in the course of the reduced evolution. The densities of these Poisson distribution will evolve in time and we derive the evolution equation for these densities (the kinetic equation). It has the form
\begin{equation}\label{eqqeeq}
  \dfrac{\partial }{\partial t}u_{t}\left( x\right) =-mu_{t}\left( x\right)
e^{-\left( u_{t}\ast \phi \right) \left( x\right) }+\uplambda.
\end{equation}
The second part of the paper (Section 5) is devoted to the detailed analysis of this non-linear non-local equation.
In Subsection 5.1 we prove the existence and uniqueness of non-negative continuous solutions to \eqref{eqqeeq}, study the stability of an equilibrium solution to this equation. Moreover, one can show that (Theorem~\ref{Thbdd}) the small enough initial data for the equation \eqref{eqqeeq} leads to the uniformly bounded solution in space and time. We also obtain a comparison principle for solutions. In Subsection 5.2, however, we show that if the initial data is large enough in some volume then the corresponding solution grows pointwise to infinity at this volume for any parameters of our system. This effect is known as the aggregation in the system. We also prove that non-local character of the equation \eqref{eqqeeq} yields an expansion of the initially localized aggregation.

\section{Basic facts and notation}\label{sect-Prelim}

Let ${\B}({\X})$ be the family of all Borel sets in ${\X}$, $d\geq
1$; ${\B}_{\mathrm{b}} ({\X})$ denotes the system of all bounded
sets from ${\B}({\X})$.

The configuration space over space $\X$ consists of all locally
finite subsets (configurations) of $\X$. Namely,
\begin{equation} \label{confspace}
\Ga =\Ga\bigl(\X\bigr) :=\Bigl\{ \ga \subset \X \Bigm| |\ga _\La
|<\infty, \ \mathrm{for \ all } \ \La \in {\B}_{\mathrm{b}}
(\X)\Bigr\}.
\end{equation}
Here $|\cdot|$ means the cardinality of a~set, and
$\ga_\La:=\ga\cap\La$. The space $\Ga$ is equipped with the vague
topology, i.e., the weakest topology for which all mappings
$\Ga\ni\ga\mapsto \sum_{x\in\ga} f(x)\in{\R}$ are continuous for any
continuous function $f$ on $\X$ with compact support. The
corresponding Borel $\sigma $-algebra $\B(\Ga )$ is the smallest
$\sigma $-algebra for which all mappings $\Ga \ni \ga \mapsto |\ga_
\La |\in{ \N}_0:={\N}\cup\{0\}$ are measurable for any $\La\in{
\B}_{\mathrm{b}}(\X)$, see e.g. \cite{AKR1998a}.

The space of $n$-point configurations in $Y\in\B(\X)$ is defined by
\begin{equation*}
\Ga^{(n)}(Y):=\Bigl\{ \eta \subset Y \Bigm| |\eta |=n\Bigr\} ,\quad
n\in { \N}.
\end{equation*}
We set $\Ga^{(0)}(Y):=\{\emptyset\}$. As a~set, $\Ga^{(n)}(Y)$ may
be identified with the symmetrization of $\widetilde{Y^n} = \bigl\{
(x_1,\ldots ,x_n)\in Y^n \bigm| x_k\neq x_l \ \mathrm{if} \ k\neq
l\bigr\}$. Hence one can introduce the corresponding Borel $\sigma
$-algebra, which we denote by $\B\bigl(\Ga^{(n)}(Y)\bigr)$. The
space of finite configurations in $Y\in\B(\X)$ is defined as
\begin{equation*}
\Ga_0(Y):=\bigsqcup_{n\in {\N}_0}\Ga^{(n)}(Y).
\end{equation*}
This space is equipped with the topology of the disjoint union. Let
$\B \bigl(\Ga_0(Y)\bigr)$ denote the corresponding Borel $\sigma
$-algebra. In the case of $Y=\X$ we will omit the index $Y$ in the
previously defined notations. Namely, $\Ga_0:=\Ga_{0}(\X)$,
$\Ga^{(n)}:=\Ga^{(n)}(\X)$.

The restriction of the Lebesgue product measure $(dx)^n$ to
$\bigl(\Ga^{(n)}, \B(\Ga^{(n)})\bigr)$ we denote by $m^{(n)}$. We
set $m^{(0)}:=\delta_{\{\emptyset\}}$. The Lebesgue--Poisson measure
$\la $ on $\Ga_0$ is defined by
\begin{equation} \label{LP-meas-def}
\la :=\sum_{n=0}^\infty \frac {1}{n!}m^{(n)}.
\end{equation}
For any $\La\in\B_{\mathrm{b}}(\X)$ the restriction of $\la$ to $\Ga
(\La):=\Ga_{0}(\La)$ will be also denoted by $\la $. The space
$\bigl( \Ga, \B(\Ga)\bigr)$ is the projective limit of the family of
spaces $\bigl\{\bigl( \Ga(\La), \B(\Ga(\La))\bigr)\bigr\}_{\La \in
\B_{\mathrm{b}} (\X)}$. The Poisson measure $\pi$ on $\bigl(\Ga
,\B(\Ga )\bigr)$ is given as the projective limit of the family of
measures $\{\pi^\La \}_{\La \in \B_{\mathrm{b}} (\X)}$, where $
\pi^\La:=e^{-m(\La)}\la $ is the probability measure on $\bigl(
\Ga(\La), \B(\Ga(\La))\bigr)$ and $m(\La)$ is the Lebesgue measure
of $\La\in \B_{\mathrm{b}} (\X)$ (see e.g. \cite{AKR1998a} for
details).

A set $M\in \B (\Ga_0)$ is called bounded if there exists $ \La \in
\B_{\mathrm{b}} (\X)$ and $N\in { \N}$ such that $M\subset
\bigsqcup_{n=0}^N\Ga^{(n)}(\La)$. The set of bounded measurable
functions with bounded support we denote by $
B_{\mathrm{bs}}(\Ga_0)$, i.e., $G\in B_{\mathrm{bs}}(\Ga_0)$ if $
G\upharpoonright_{\Ga_0\setminus M}=0$ for some bounded $M\in {\B
}(\Ga_0)$. Any $\B(\Ga_0)$-measurable function $G$ on $ \Ga_0$, in
fact, is defined by a~sequence of functions
$\bigl\{G^{(n)}\bigr\}_{n\in{ \N}_0}$ where $G^{(n)}$ is a
$\B(\Ga^{(n)})$-measurable function on $\Ga^{(n)}$. The set of
\textit{cylinder functions} on $\Ga$ we denote by ${{\mathcal{
F}}_{\mathrm{cyl}}}(\Ga )$. Each $F\in
{{\mathcal{F}}_{\mathrm{cyl}}}(\Ga )$ is characterized by the
following relation: $F(\ga )=F(\ga_\La )$ for some $\La\in
\B_{\mathrm{b}}(\X)$. Functions on $\Ga$ will be called {\em
observables} whereas functions on $\Ga_0$ well be called {\em
quasi-observables}.

There exists mapping from $B_{\mathrm{bs}} (\Ga_0)$ into ${{
\mathcal{F}}_{\mathrm{cyl}}}(\Ga )$, which plays the key role in our
further considerations:
\begin{equation}
(KG)(\ga ):=\sum_{\eta \Subset \ga }G(\eta ), \quad \ga \in \Ga,
\label{KT3.15}
\end{equation}
where $G\in B_{\mathrm{bs}}(\Ga_0)$, see e.g.
\cite{KK2002,Len1975,Len1975a}. The summation in \eqref{KT3.15} is
taken over all finite subconfigurations $\eta\in\Ga_0$ of the
(infinite) configuration $\ga\in\Ga$; we denote this by the symbol,
$\eta\Subset\ga $. The mapping $K$ is linear, positivity preserving,
and invertible, with
\begin{equation}
(K^{-1}F)(\eta ):=\sum_{\xi \subset \eta }(-1)^{|\eta \setminus \xi
|}F(\xi ),\quad \eta \in \Ga_0. \label{k-1trans}
\end{equation}
\begin{remark}\label{rem:finonfin}
We would like to stress that the right hand side of \eqref{k-1trans}
is well-defined for any $F$ which is pointwise defined at least on
the whole $\Ga_0$.
\end{remark}

The so-called coherent state corresponding to a~$\B(\X)$-measurable
function $f$ is defined by
\begin{equation}\label{LPexp}
e_\la (f,\eta ):=\prod_{x\in \eta }f(x) ,\ \eta \in \Ga
_0\!\setminus\!\{\emptyset\},\quad e_\la (f,\emptyset ):=1.
\end{equation}Then
\begin{equation}\label{Kexp}
(Ke_\la (f))(\eta)=e_\la(f+1,\eta), \quad \eta\in\Ga_0
\end{equation}
and for any $f\in L^1(\X,dx)$
\begin{equation}\label{intexp}
\int_{\Ga_0}e_\la (f,\eta)d\la(\eta)=\exp\Bigl\{\int_\X
f(x)dx\Bigr\}.
\end{equation}

A measure $\mu \in {\mathcal{M}}_{\mathrm{fm} }^1(\Ga )$ is called
locally absolutely continuous with respect to the Poisson measure
$\pi$ if for any $\La \in \B_{\mathrm{b}} (\X)$ the projection of
$\mu$ onto $\Ga(\La)$ is absolutely continuous with respect to the
projection of $ \pi$ onto $\Ga(\La)$. In this case, according to
\cite{KK2002}, there exists a~\emph{correlation functional}
$k_{\mu}:\Ga_0 \rightarrow [0,+\infty)=:{\R}_+$ such that for any $G\in
B_{\mathrm{bs}} (\Ga_0)$ the following equality holds
\begin{equation} \label{eqmeans}
\int_\Ga (KG)(\ga) d\mu(\ga)=\int_{\Ga_0}G(\eta)
k_\mu(\eta)d\la(\eta).
\end{equation}
The functions $ k_{\mu}^{(n)}:(\R^{d})^{n}\longrightarrow\R_{+} $
given by
\[
k_{\mu}^{(n)}(x_{1},\ldots,x_{n}):=
\begin{cases}
k_{\mu}(\{x_{1},\ldots,x_{n}\}), & \mathrm{if} \
(x_{1},\ldots,x_{n})\in \widetilde{(\R^{d})^{n}}\\
0, & \mathrm{ otherwise}
\end{cases}
\]
are called \emph{correlation functions} of the measure $\mu$. Note
that $k_\mu^{(0)}=1$.

Below we would like to mention without proof the partial case of the
well-known technical lemma (see e.g. \cite{KMZ2004}) which plays
very important role in our calculations.

\begin{lemma}
\label{Minlos} For any measurable function $H:\Ga_0\times\Ga_0\times
\Ga_0\rightarrow{\R}$
\begin{equation} \label{minlosid}
\int_{\Ga _{0}}\sum_{\xi \subset \eta }H\left( \xi ,\eta \setminus
\xi ,\eta \right) d\la \left( \eta \right) =\int_{\Ga _{0}}\int_{\Ga
_{0}}H\left( \xi ,\eta ,\eta \cup \xi \right) d\la \left( \xi
\right) d\la \left( \eta \right)
\end{equation}
if both sides of the equality make sense.
\end{lemma}

\section{The microscopic description of the model}

In the present paper we consider an infinite system of undistinguishable particles in $\R^{d}$ evolving in time. The mechanism of the evolution includes birth with the constant rate and the density dependent death. It can be described by the following heuristic generator
\begin{multline}\label{gen}
\left( LF\right) \left( \gamma \right) :=m\sum_{x\in \gamma }e^{-E^{\phi
}\left( x,\gamma \setminus \{x\}\right) }\left[ F\left( \gamma \setminus
\{x\}\right) -F\left( \gamma \right) \right] \\+\uplambda \int_{\mathbb{R}^{d}}\left[
F\left( \gamma \cup \{x\}\right) -F\left( \gamma \right) \right] dx.
\end{multline}
%for any $F\in K\bigl(\Bbs\bigr)$ (we will subsequently use the notation $x$
%instead of the more precise~$\{x\}$).
Here $m>0$ is a mortality rate, $\uplambda >0$ is an immigration (birth) rate, and
\begin{equation}\label{energy}
E^\phi(x,\ga\setminus x):=\sum_{y\in\ga\setminus x}\phi(x-y)\in[0,+\infty],
\end{equation}
where $\phi:\X\to\R_+$ is such
that $\phi(-x)=\phi(x)$, $x\in\X$ and
\begin{equation}\label{exp-int}
C_\phi:=\int_\X \bigl(1-e^{-\phi(x)}\bigr)\,dx<\infty.
\end{equation}
By abuse of notation we continue to write $x$ for $\{x\}$ when no confusion can arise.
It is worth pointing out that the death rate is a decreasing function in $\ga$. Namely, $\ga^{\prime}\subset\ga$ yields
\begin{equation}\label{deathr}
d(x,\,\ga\setminus x):=me^{-E^{\phi}(x,\,\ga\setminus x)}\leq d(x,\,\ga^{\prime}\setminus x)
\end{equation}
Another way of stating \eqref{deathr} is to say that probability to die for an element of a configuration $\ga$ becomes lower if this element is densely surrounded by other elements of $\ga$.

In order to define all objects in the formula \eqref{gen} rigorously it is necessary to put some restrictions on the class of functions $F$. Let $F\in K\bigl(\Bbs\bigr)$. Since $K\bigl(\Bbs\bigr)\subset\Fcyl$ the integral in \eqref{gen}
is taking over a bounded domain
%, say, $\La\in\Bb$,
and, hence, it is finite. By the same
argument, the sum in \eqref{gen} is taking over a finite configuration.
%$\ga_\La$ and putting $e^{-\infty}:=0$ we obtain that \eqref{gen}
%is well-defined pointwise, for any $\ga\in\Ga$.
We adhere to the convention that $d(x,\ga)=0$ if $E^{\phi}=\infty$. As a result the expression \eqref{gen} is pointwise well-defined.
In particular, the right hand side of \eqref{gen} is
finite, for any $\ga\in\Ga_0$.

Our next goal is to derive the evolution equation for correlation functions determined by the mechanism  \eqref{gen}. The corresponding scheme was proposed, e.g., in \cite{KKM09,FKO2009,FKK2011a}. Below we introduce and study objects which are necessary for the construction of evolution of correlation function as well as for the mesoscopic description of the system under consideration.

First, we have to define the symbol operator
\begin{equation}\label{descdef}
(\widehat{L}G)(\eta):=(K^{-1}LKG)(\eta),\quad \eta\in\Ga_0,
\end{equation}
corresponding to the operator $L$. By Remark~\ref{rem:finonfin}, the expression at the right hand side of \eqref{descdef} will have sense if we define it
pointwise for any $G\in\Bbs$.
\begin{proposition}\label{symgen}
For any $G\in\Bbs$, $\eta\in\Ga_0$,
\begin{align}\notag
(\hat{L}G)\left( \eta \right) =&-m\sum_{\xi \subset \eta }G(\xi )\sum_{x\in
\xi }e^{-E^{\phi }\left( x,\xi \setminus x\right) }e_{\lambda }\bigl(
e^{-\phi \left( x-\cdot \right) }-1,\eta \setminus \xi \bigr)\\
&+\uplambda \int_{\mathbb{R}^{d}}G\left( \eta \cup x\right) dx.\label{descgen}
\end{align}
\end{proposition}
\begin{proof}
As was shown in \cite[Proposition~3.1]{FKK2011a}, for any $G\in\Bbs$,
$\eta\in\Ga_0$,
\begin{align} (\hat{L}G)(\eta ) =&-\sum_{\xi \subset \eta
}G(\xi )\sum_{x\in \xi }\bigl(K^{-1}d(x,\cdot\cup\xi\setminus
x)\bigr)(\eta\setminus
\xi)\notag\\
&+\sum_{\xi \subset \eta }\int_{\R^{d}}\,G(\xi \cup
x)\bigl(K^{-1}b(x,\cdot\cup\xi)\bigr)(\eta\setminus \xi) dx, \label{newexpr}
\end{align}
where $d(x,\ga)=me^{-E^\phi(x,\ga)}$, $b(x,\ga)=\uplambda $, $x\notin\ga$.
By \eqref{LPexp} and \eqref{Kexp}, one has
\begin{align*}
\bigl(K^{-1}d(x,\cdot\cup\xi\setminus
x)\bigr)(\eta\setminus
\xi)&=me^{-E^\phi(x,\xi\setminus x)}\bigl(K^{-1}e^{-E^\phi(x,\cdot)}\bigr)
(\eta\setminus\xi)\\&=me^{-E^\phi(x,\xi\setminus x)}e_\la\bigl(e^{-\phi(x-\cdot)}-1,\eta\setminus
\xi\bigr),
\end{align*}
and, clearly,
\[
\bigl(K^{-1}b(x,\cdot\cup\xi)\bigr)(\eta\setminus \xi)=\uplambda 0^{|\eta\setminus\xi|}=\begin{cases}\uplambda ,&\xi=\eta,\\0,&\text{otherwise}.\end{cases}
\]
The statement is proved.
\end{proof}

Let us mention one important consequence of the previous proposition. It is easy to check that both sides of the equality  \eqref{minlosid} are finite for the functions
$$
H_{1}(\xi,\eta,\eta\cup\xi)=-mG(\xi)\sum_{x\in\xi}e^{E^{\phi}(x,\xi\setminus x)}e_{\la}(e^{-\phi(x-\cdot)}-1, \eta )
k(\eta\cup x)
$$
and
$$
H_{2}(\xi,\eta,\eta\cup\xi)=\uplambda1\!\!1_{\{x\}}(\xi)k(\eta)G(\eta\cup \xi),
$$
where $k\in B_{bs}(\Ga_{0})$ is an arbitrary function. By Lemma \ref{Minlos}, Propositon \ref{symgen}, and the assumption \eqref{exp-int} it may be concluded that
\begin{equation}\label{operdual}
\int_{\Ga_0}(\widehat{L}G)(\eta) k(\eta)\,d\la(\eta)=\int_{\Ga_0}G(\eta)(L^\triangle
k)(\eta) \,d\la(\eta),
\end{equation}
where the mapping $L^\triangle k$ is given by
%By Lemma~\ref{Minlos}, we obtain from \eqref{descgen} and \eqref{operdual} that, say, for any $k\in\Bbs$, $\eta \in\Ga_0$,
\begin{align}\notag
\left( L^{\Delta }k\right) \left( \eta \right) =&-m\sum_{x\in \eta
}e^{-E^{\phi }\left( x,\eta \setminus x\right) }\int_{\Gamma _{0}}k\left(
\eta \cup \xi \right) e_{\lambda }\bigl( e^{-\phi \left( x-\cdot \right)
}-1,\xi \bigr) d\lambda \left( \xi \right) \\&+\uplambda \sum_{x\in \eta }k\left( \eta
\setminus x\right).\label{Ltriangle}
\end{align}

For any $C>0$, we consider the following Banach space of functions

\begin{equation*}
{\K}_{C}:=\left\{ k:\Ga _{0}\rightarrow {\R}\,\Bigm| k\cdot
C^{-|\cdot |}\in L^{\infty }(\Ga _{0},\la )\right\}
\end{equation*} with the norm
\[
\Vert k\Vert _{{\K}_{C}}:=\Vert C^{-|\cdot |}k(\cdot )\Vert
_{L^{\infty }(\Ga _{0},\la )}.
\]
Clearly,
\begin{equation}\label{RB}
|k(\eta)|\leq \Vert k\Vert_{\K_C}C^{|\eta|}, \quad \text{for $\la$-a.a.}\
 \eta\in\Ga_0.
\end{equation}
It is easily seen, that
\begin{equation}\label{scaleofspaces}
\K_C\subset\K_{C'}, \qquad \|\cdot\|_{\K_C}\geq \|\cdot\|_{\K_{C'}},\qquad C\leq C',
\end{equation}
and hence, one has a continuum family of embedded Banach spaces $\K_C$. We will deal with the
operator \eqref{Ltriangle} acting in this family from any $\K_C$ to any larger space $\K_{C'}$, with $C'>C$.

\begin{proposition}\label{normest}
Suppose that \eqref{exp-int} holds.
Let $C>C_{0}>0$ be arbitrary. Then, for any $C^{\prime },C^{\prime \prime }$
such that $C_{0}\leq C^{\prime }<C^{\prime \prime }\leq C$, the mapping \eqref{Ltriangle} defines a linear operator which acts from $\K_{C'}$ to $\K_{C''}$. Moreover, for any $%
k\in \mathcal{K}_{C^{\prime }}$,
\begin{equation}\label{normoper}
\bigl\Vert L^\triangle k\bigr\Vert_{\K_{C''}}\leq \frac{1}{C^{\prime \prime
}-C^{\prime }} \frac{C}{e}%
\left( me^{CC_{\phi }}+\frac{\uplambda }{C_{0}}\right) \left\Vert k\right\Vert _{\mathcal{K}_{C^{\prime }}}.
\end{equation}
\end{proposition}
\begin{proof} Using \eqref{Ltriangle}, one has, for any $k\in\K_{C'}$,
$\eta\in\Ga_0$,
\begin{align*}
&\quad\left( C^{\prime \prime }\right) ^{-\left\vert \eta \right\vert
}\left\vert \bigl( L^{\Delta }k\bigr) \left( \eta \right) \right\vert \\
&\leq \left( C^{\prime \prime }\right) ^{-\left\vert \eta \right\vert
}m\sum_{x\in \eta }e^{-E^{\phi }\left( x,\eta \setminus x\right)
}\int_{\Gamma _{0}}\left\vert k\left( \eta \cup \xi \right) \right\vert
e_{\lambda }\left( \left\vert e^{-\phi \left( x-\cdot \right) }-1\right\vert
,\xi \right) d\lambda \left( \xi \right) \\&\quad+\uplambda \left( C^{\prime \prime }\right)
^{-\left\vert \eta \right\vert }\sum_{x\in \eta }\left\vert k\left( \eta
\setminus x\right) \right\vert \\\intertext{and, since $\phi\geq0$
and \eqref{RB} holds, one can continue}
&\leq \left( C^{\prime \prime }\right) ^{-\left\vert \eta \right\vert
}m\sum_{x\in \eta }\int_{\Gamma _{0}}\left( C^{\prime }\right) ^{\left\vert
\eta \right\vert +\left\vert \xi \right\vert }\left\Vert k\right\Vert _{\mathcal{K}_{C^{\prime }}}e_{\lambda }\left( 1- e^{-\phi \left(
x-\cdot \right) },\xi \right) d\lambda \left( \xi \right)
\\&\quad+\uplambda \left( C^{\prime \prime }\right) ^{-\left\vert \eta \right\vert
}\sum_{x\in \eta }\left( C^{\prime }\right) ^{\left\vert \eta \right\vert
-1}\left\Vert k\right\Vert _{\mathcal{K}_{C^{\prime }}} \\\intertext{and,
finally, using \eqref{intexp} and \eqref{exp-int}, the latter expression is equal to}
&\left\Vert k\right\Vert _{\mathcal{K}_{C^{\prime }}}\left( \frac{%
C^{\prime }}{C^{\prime \prime }}\right) ^{\left\vert \eta \right\vert
}\left\vert \eta \right\vert \left( me^{C^{\prime }C_{\phi }}+\frac{\uplambda }{%
C^{\prime }}\right) \\&\leq\left\Vert k\right\Vert _{\mathcal{K}_{C^{\prime }}}\left( \frac{%
C^{\prime }}{C^{\prime \prime }}\right) ^{\left\vert \eta \right\vert
}\left\vert \eta \right\vert \left( me^{CC_{\phi }}+\frac{\uplambda }{%
C_0}\right).
\end{align*}%
Since, for any $a\in \left( 0,1\right)$,
\begin{equation}\label{max1}
\max_{t\geq 0}\left( ta^{t}\right) =\frac{1}{e\left( -\ln a\right) },
\end{equation}
it follows that%
\[
\left( C^{\prime \prime }\right) ^{-\left\vert \eta \right\vert }\left\vert
\left( L^{\Delta }k\right) \left( \eta \right) \right\vert \leq \left\Vert
k\right\Vert _{\mathcal{K}_{C^{\prime }}}\frac{1}{e\left( \ln C^{\prime
\prime }-\ln C^{\prime }\right) }\left( e^{CC_{\phi }}+\frac{\uplambda }{C_{0}}%
\right).
\]%
But there exists $c\in \left[ C^{\prime },C^{\prime \prime }\right] \subset %
\left[ C_{0},C\right] $ such that%
\[
\ln C^{\prime \prime }-\ln C^{\prime }=\frac{1}{c}\left( C^{\prime \prime
}-C^{\prime }\right) \geq \frac{1}{C}\left( C^{\prime \prime }-C^{\prime
}\right) >0,
\]
which proves the proposition.
\end{proof}

Our next concern will be to study the evolution equation
\begin{equation}\label{coreq}
\frac{\partial}{\partial t}k_t(\eta)=(L^\triangle k_{t})(\eta),
\end{equation}
$$
k_{t}|_{t=0}=k_{0}
$$
in the scale of Banach spaces \eqref{scaleofspaces}. The important point to note here is that a solution to this equation describes the evolution of correlation functions of the system. Roughly, speaking by solving this equation we construct statistical dynamics of states corresponding to the mechanism of evolution described by the operator $L$ and a given initial distribution with correlation function $k_{0}$. To study equation \eqref{coreq}  we use the following approach which goes back to the
the Picard-type approximations and a method by A.~G.~Kostyuchenko and G.~E.~Shilov presented in \cite[Appendix 2, A2.1]{GS1967}. This method, originally considered for
equations with time independent coefficients, has been extended to an abstract
and general framework by T.~Yamanaka in \cite{Yam1960} and L.~V.~Ovsjannikov in
\cite{Ovs1965} in the linear case, and many applications were exposed by F.~Treves in \cite{Tre1968} who probably initiated the terminology
{\em Ovsjannikov theorem}. For the convenience of the reader we give below one of the versions of this theorem which may be found, e.g., in~\cite[Theorem 2.5]{FKO2011a}).

\begin{theorem}\label{propOvs}
Consider one-parameter decreasing family of Banach spaces $\{\mathbb{B}_s \mid 0 < s \leq s_0\}$ such that
$B_{s''}\subset B_{s'}$, $\Vert\cdot\Vert_{s'}\leq\Vert\cdot\Vert_{s''}$,
for any pair $s',s''$ such that $0 < s'< s''\leq s_0$, where $\Vert\cdot\Vert_s$ denotes the norm in $\mathbb{B}_s$. Consider also the initial
value problem
\begin{equation}\label{CPa}
\frac{du(t)}{dt}= Au(t), \quad u(0) = u_0\in\mathbb{B}_{s_0},
\end{equation}
where, for each fixed $s\in(0,s_0)$ and for each pair $s',s''$ such that $s\leq s'<s''\leq s_0$,
the mapping $A : \mathbb{B}_{s''} \to \mathbb{B}_{s'}$ is linear satisfying
\[
\Vert Au\Vert_{s'} \leq
\frac{M}{s''- s'}\Vert u\Vert_{s''},
\]
for some $M > 0$ and all $u\in\mathbb{B}_{s''}$. The constant $M$ is independent of $s',s''$ and $u$, however it might depend continuously on $s,s_0$.

Then, for each $s \in (0,s_0)$, there is a constant $\delta=(eM)^{-1} > 0$ and a unique function $u : \bigl[0,\delta(s_0- s)\bigr)\to
\mathbb{B}_s$ which is continuously differentiable on $\bigl(0,\delta (s_0- s)\bigr)$ in $\mathbb{B}_s$, $Au \in \mathbb{B}_s$, and $u$
solves \eqref{CPa} on the time-interval $0 \leq t < \delta(s_0- s)$.
\end{theorem}

The application of this theorem to our situation is stated below. It says that evolution of correlation functions corresponding to the system under consideration exists on a finite time interval.
\begin{proposition} \label{100502}
Suppose that \eqref{exp-int} holds. Let $C_{0}>0$ be arbitrary and fixed. Consider the initial value problem \eqref{coreq} with  $k_{0}\in\K_{C_0}$. Then, for each $C>C_{0}$, there is a time $0<T(C_{0},C)\leq(1+2\sqrt{\uplambda C_\phi})^{-1}$ and a unique function $k : \bigl[0,T(C_{0},C)\bigr)\to
\K_{C}$ which is continuously differentiable on $\bigl(0,T(C_{0},C)\bigr)$ in $\K_{C}$, and $k$
solves \eqref{coreq} on the time-interval $0 \leq t < T(C_{0},C)$.
\end{proposition}
\begin{proof} We shall have established the proposition if we apply Theorem \ref{propOvs} to the following
family
\begin{equation}\label{newfamily}
\mathbb{B}_s:=\K_{\frac{1}{s}}, \quad s>0,
\end{equation}
with initial index $s_0:=\dfrac{1}{C_0}$. In this case, one can rewrite \eqref{normoper} for
\[
s'':=\dfrac{1}{C''}<\dfrac{1}{C'}=:s', \quad s=\dfrac{1}{C},
\]
in the following form
\begin{align}\notag
\bigl\Vert L^\triangle k\bigr\Vert_{\mathbb{B}_{s''}}&\leq \frac{s's''}{s'-s''} \frac{1}{se}%
\left( me^{\frac{C_{\phi }}{s}}+\uplambda s_0 \right) \left\Vert k\right\Vert _{\mathbb{B}_{s'}}\\
&\leq \frac{1}{s'-s''} \frac{s_0^2}{se}%
\left( me^{\frac{C_{\phi }}{s}}+\uplambda s_0 \right) \left\Vert k\right\Vert _{\mathbb{B}_{s'}}\label{normoper-rewr}\\
&=:\frac{M(s,s_{0})}{s'-s''}\left\Vert k\right\Vert _{\mathbb{B}_{s'}}.\notag
\end{align}
Then, by Theorem~\ref{propOvs}, we obtain an evolution $\mathbb{B}_{s_0}\ni
k_0\mapsto k_t\in\mathbb{B}_s$, for all times which are less then
\begin{align}\notag
\bigl(eM(s,s_{0})\bigr)^{-1}(s_0-s)&=\biggl(\frac{s_0^2}{s}%
\Bigl( me^{\frac{C_{\phi }}{s}}+\uplambda s_0 \Bigr)\biggr)^{-1}(s_0-s)\\ &=\frac{C_0^2}{C\bigl(e^{CC_\phi}+\frac{\uplambda }{C_0}\bigr)}\Bigl(\frac{1}{C_0}-\frac{1}{C}\Bigr)\notag\\&=\frac{C_{0}(C-C_{0})}{C
^{2}\bigl( e^{CC_{\phi }}+\frac{\uplambda }{C_{0}}\bigr) }=:T(C_{0},C).\label{T}
\end{align}
It should be noted that
\[
T(C_{0},C)\leq \frac{C_{0}C}{C
^{2}\bigl( e^{CC_{\phi }}+\frac{\uplambda }{C_{0}}\bigr) }\leq \frac{1}{ e^{CC_{\phi }}+\frac{\uplambda }{C} }< \frac{1}{ 1+CC_{\phi }+\frac{\uplambda }{C} }\leq\frac{1}{1+2\sqrt{\uplambda C_\phi}},
\]
and the proposition follows.
\end{proof}
%Summarizing, we have thus shown that  of correlation functions corresponding to the  mechanism of evolution $L$ on the finite time interval
%%\[
%%\mathcal{K}_{C_{0}}\ni k_{0}\mapsto
%%k_{t}\in \mathcal{K}_{C}, \quad t\in
%
%$\lbrack 0,T(C_0,C)$ is given by \eqref{T}.

\section{The mesoscopic description of the model}

In this section we will be concerned with the mesoscopic description of our microscopic model. To this end we implement the general scheme proposed in \cite{FKK2010a} (see also \cite{FKK2011a,FKK2010b,FKKoz2011}). The basic idea of this scheme is to scale the microscopic system making all interactions weak and the density of the system appropriately high.
We are thus led to the following formal steps. We rescale the operator \eqref{gen}:
\begin{align}\notag
\left( L_\eps F\right) \left( \gamma \right) :=&m\sum_{x\in \gamma }e^{-\eps E^{\phi
}\left( x,\gamma \setminus x\right) }\left[ F\left( \gamma \setminus
x\right) -F\left( \gamma \right) \right] \\&+\frac{\uplambda }{\eps}\int_{\mathbb{R}^{d}}\left[
F\left( \gamma \cup x\right) -F\left( \gamma \right) \right] dx.\label{scaleofgen}
\end{align}
We continue in the same fashion as before considering  $\widehat{L}_\eps:=K^{-1}L_\eps
K$ and, by analogy to \eqref{operdual}, the dual operator $L_\eps^\triangle$.
According to the general scheme we introduce the following renormalization of the operator $L_\eps^\triangle$:
\begin{equation}
L_{\varepsilon ,\mathrm{ren}}^{\Delta }:=R_\eps L_{\varepsilon }^{\Delta }R_{\eps^{-1}},
\end{equation}
where, for any $c>0$,
\[
(R_ck)(\eta):=c^{|\eta|}k(\eta), \quad\eta\in\Ga_0.
\]
Roughly speaking, the renormalization procedure means the following: we increase the density of the system (action of $R_{\eps^{-1}}$) making all interactions weaker (action of the operator $L_{\varepsilon }^{\Delta }$) and, afterwards, we return to the initial level of the density (action of $R_{\eps}$).

The precise form of the operator $L_{\varepsilon ,\mathrm{ren}}^{\Delta }$ is established by our next proposition.
\begin{proposition}
For any $k\in\Bbs$, $\eta\in\Ga_0$,
\begin{align}\notag
\left( L_{\varepsilon ,\mathrm{ren}}^{\Delta }k\right) \left( \eta \right)
=&-m\sum_{x\in \eta }e^{-\varepsilon E^{\phi }\left( x,\eta \setminus
x\right) }\int_{\Gamma _{0}}k\left( \eta \cup \xi \right) e_{\lambda }\left(
\frac{e^{-\varepsilon \phi \left( x-\cdot \right) }-1}{\varepsilon },\xi
\right) d\lambda \left( \xi \right) \\&+\uplambda \sum_{x\in \eta }k\left( \eta
\setminus x\right).\label{Ltriangleeps}
\end{align}%
\end{proposition}
\begin{proof}
By \eqref{Ltriangle}, we obviously get
\begin{align*}
\left( L_{\varepsilon }^{\Delta }k\right) \left( \eta \right) =&-\sum_{x\in
\eta }e^{-\varepsilon E^{\phi }\left( x,\eta \setminus x\right)
}\int_{\Gamma _{0}}k\left( \eta \cup \xi \right) e_{\lambda }
\bigl(
e^{-\varepsilon \phi \left( x-\cdot \right) }-1,\xi \bigr) d\lambda \left(
\xi \right) \\&+\frac{\uplambda }{\varepsilon }\sum_{x\in \eta }k\left( \eta \setminus
x\right).
\end{align*}
Therefore, one has
\begin{align*}
&\left( L_{\varepsilon ,\mathrm{ren}}^{\Delta }k\right) \left( \eta \right)
\\
=&-m\varepsilon ^{\left\vert \eta \right\vert }\sum_{x\in \eta
}e^{-\varepsilon E^{\phi }\left( x,\eta \setminus x\right) }\int_{\Gamma
_{0}}\varepsilon ^{-\left\vert \eta \right\vert }\varepsilon ^{-\left\vert \xi \right\vert }k\left( \eta \cup \xi
\right) e_{\lambda }\left(
e^{-\varepsilon \phi \left( x-\cdot \right) }-1,\xi \right) d\lambda \left(
\xi \right) \\&+\varepsilon ^{\left\vert \eta \right\vert }\frac{\uplambda }{\varepsilon
}\sum_{x\in \eta }\varepsilon ^{-\left\vert \eta \right\vert +1}k\left( \eta
\setminus x\right),
\end{align*}%
which is our desired conclusion.
\end{proof}

Having disposed of this preliminary step, we can now look to the pointwise limit of  $L_{\varepsilon ,\mathrm{ren}}^{\Delta }$. It is easy to check that
\begin{align}\label{defLV}
\left( L_{V}^{\Delta }k\right) \left( \eta \right) :=&\lim_{\eps\to0}\left( L_{\varepsilon ,\mathrm{ren}}^{\Delta }k\right) \left( \eta \right)\\=&-m\sum_{x\in \eta}\int_{\Gamma _{0}}k\left( \eta \cup \xi \right) e_{\lambda }\left( -\phi
\left( x-\cdot \right) ,\xi \right) d\lambda \left( \xi \right)\notag \\&+\uplambda \sum_{x\in
\eta }k\left( \eta \setminus x\right).\label{LtriangleV}
\end{align}

From now on, we always suppose that $\phi\in L^1(\X)$, namely,
\begin{equation}\label{intcond}
\beta:=\int_\X \phi(x) \,dx<\infty.
\end{equation}
We see at once that \eqref{intcond} implies \eqref{exp-int}, and
\[
C_\phi\leq\beta.
\]
We next modify Proposition \ref{normest} to the case of the operators $L_{\varepsilon ,\mathrm{ren}}^{\Delta }$ and $L_{V}^{\Delta }$ under assumption \eqref{intcond}.

\begin{proposition}\label{normestscale}
Suppose that \eqref{intcond} holds.
Let $C>C_{0}>0$ be arbitrary. Then, for any $C^{\prime },C^{\prime \prime }$
such that $C_{0}\leq C^{\prime }<C^{\prime \prime }\leq C$, and for any $%
k\in \mathcal{K}_{C^{\prime }}$, one has
\begin{equation}\label{normopersharp}
\bigl\Vert L^\triangle_\sharp k\bigr\Vert_{\K_{C''}}\leq \frac{1}{C^{\prime \prime
}-C^{\prime }} \frac{C}{e}%
\left( me^{C\beta}+\frac{\uplambda }{C_{0}}\right) \left\Vert k\right\Vert _{\mathcal{K}_{C^{\prime }}},
\end{equation}
where $L^\triangle_\sharp$ means $L^\triangle_{\eps,\mathrm{ren}}$
or $L^\triangle_V$.
\end{proposition}
\begin{proof}
By \eqref{Ltriangleeps}, one has
\begin{align*}
&\left( C^{\prime \prime }\right) ^{-\left\vert \eta \right\vert
}\left\vert \left( L_{\varepsilon ,\mathrm{ren}}^{\Delta }k\right) \left(
\eta \right) \right\vert \\
\leq &\left( C^{\prime \prime }\right) ^{-\left\vert \eta \right\vert
}\sum_{x\in \eta }e^{-\varepsilon E^{\phi }\left( x,\eta \setminus x\right)
}\int_{\Gamma _{0}}\left\vert k\left( \eta \cup \xi \right) \right\vert
e_{\lambda }\left( \left\vert \frac{e^{-\varepsilon \phi \left( x-\cdot
\right) }-1}{\varepsilon }\right\vert ,\xi \right) d\lambda \left( \xi
\right) \\&+\uplambda \left( C^{\prime \prime }\right) ^{-\left\vert \eta \right\vert
}\sum_{x\in \eta }\left\vert k\left( \eta \setminus x\right) \right\vert. \end{align*}
Then, using inequalities $\phi\geq0$ and
\begin{equation}\label{obvest}
\left\vert \frac{e^{-\varepsilon \phi \left( y
\right) }-1}{\varepsilon }\right\vert =\frac{1-e^{-\varepsilon \phi \left( y
\right) }}{\varepsilon } \leq\phi(y), \quad y\in\X,
\end{equation}
we obtain
\begin{align*}
&\left( C^{\prime \prime }\right) ^{-\left\vert \eta \right\vert
}\left\vert \left( L_{\sharp}^{\Delta }k\right) \left(
\eta \right) \right\vert \\
\leq &\left( C^{\prime \prime }\right) ^{-\left\vert \eta \right\vert
}m\sum_{x\in \eta }\int_{\Gamma _{0}}\left( C^{\prime }\right) ^{\left\vert
\eta \right\vert +\left\vert \xi \right\vert }\left\Vert k\right\Vert _{%
\mathcal{K}_{C^{\prime }}}e_{\lambda }\left( \phi \left( x-\cdot \right)
,\xi \right) d\lambda \left( \xi \right) \\&+\uplambda \left( C^{\prime \prime }\right)
^{-\left\vert \eta \right\vert }\sum_{x\in \eta }\left( C^{\prime }\right)
^{\left\vert \eta \right\vert -1}\left\Vert k\right\Vert _{\mathcal{K}%
_{C^{\prime }}} \\\intertext{and, using \eqref{intexp}, one has}
=&\left\Vert k\right\Vert _{\mathcal{K}_{C^{\prime }}}\left( \frac{%
C^{\prime }}{C^{\prime \prime }}\right) ^{\left\vert \eta \right\vert
}\left\vert \eta \right\vert \left( e^{C^{\prime }\beta }+\frac{\uplambda }{C^{\prime
}}\right).
\end{align*}
The rest of the proof runs as in Proposition~\ref{normest}.
\end{proof}

Application of Proposition~\ref{propOvs} similar to that in the proof of Proposition \ref{100502} leads us to the following result, the detailed verification of which being left to the reader.
\begin{proposition}
Suppose that \eqref{intcond} holds. Let $C_{0}>0$ be arbitrary and fixed. Consider the initial value problems
\begin{equation}\label{CPeps}
\frac{\partial}{\partial t}k_{t,\eps}=L^\triangle_{\eps,\mathrm{ren}} k_{t,\eps}, \quad k_{t,\eps}|_{t=0}=k_{0,\eps}, \quad \eps>0,
\end{equation}
and
\begin{equation}\label{CPVl}
\frac{\partial}{\partial t}k_{t,V}=L^\triangle_{V} k_{t,V},\quad k_{t,V}|_{t=0}=k_{0,V}.
\end{equation}
with  $\bigl\{k_{0,\eps}, k_{0,V}\bigr\}_{\eps>0}\subset \K_{C_0}$. Then for each $C>C_{0}$, there exists a time
\begin{equation}\label{T1}
T_{1}=T_1(C_{0},C):=\frac{C_{0}(C-C_{0})}{C
^{2}\bigl( e^{\beta C}+\frac{\uplambda }{C_{0}}\bigr) }.
\end{equation}
and unique functions $k_{\eps},\, k_{V} : \bigl[0,T_{1}\bigr)\to
\K_{C}$, $\eps>0$ which are continuously differentiable on $\bigl(0,T_{1}\bigr)$ in $\K_{C}$, and they
solve \eqref{CPVl} and \eqref{T1}, respectively, on the time-interval $0 \leq t < T_{1}$.
\end{proposition}
\begin{remark}
It is worth noting that $T_1<T$, where $T=T(C_0,C)$ is given by \eqref{T}.
\end{remark}

Having in mind \eqref{defLV}, it is of interest to know whether solutions for \eqref{CPeps} converge to the solution for \eqref{CPVl} on the interval $[0,T_1)$ as $\eps$ tends to $0$.
To shed some light on this problem we will use the following result presented in \cite[Theorem~4.3]{FKO2011a}. For the convenience of the reader we repeat the statement of this theorem below.
\begin{proposition}\label{Thconv} Let the family of Banach spaces
$\{\B_s: 0<s\leq s_0\}$ be such as in Proposition~\ref{propOvs}. Consider a family of initial value problems
\begin{equation}
\frac{du_\eps(t)}{dt}=A_{\eps}u_\eps(t),\quad u_\eps(0)=u_{\eps}\in
\mathbb{B}_{s_0},\quad \eps\geq0,\label{V1eps}
\end{equation}
where, for each $s\in(0,s_0)$ fixed and for each pair $s', s''$ such that
$s\leq s'<s''\leq s_0$, $A_\eps:\B_{s''}\to\B_{s'}$ is a linear mapping so that
there is an $M>0$ such that for all $u\in\B_{s''}$
\begin{equation}\label{normest2}
\|A_{\eps}u\|_{s'}\leq\frac{M}{s''-s'}\|u\|_{s''}.
\end{equation}
Here $M$ is independent of $\eps, s',s''$ and $u$, however it might depend
continuously on $s,s_0$. Assume that there is a $p\in\N$ and for each
$\eps>0$ there is an $N_\eps>0$ such that for each pair $s', s''$,
$s\leq s'<s''\leq s_0$, and all $u\in\B_{s''}$
\begin{equation}\label{estdif}
\|A_{\eps}u-A_0u\|_{s'}\leq \sum_{k=1}^p\frac{N_\eps}{(s''-s')^k}\|u\|_{s''}.
\end{equation}
In addition, assume that $\lim_{\eps\rightarrow0}N_\eps=0$ and
$\lim_{\eps\rightarrow0}\|u_{\eps}(0)-u_{0}(0)\|_{s_0}=0$.

Then, for each $s\in(0,s_0)$, there is a constant $\delta=(eM)^{-1}>0$ such that there is a unique solution
$u_\eps:\left[0,\delta(s_0-s)\right)\to\B_s$, $\eps\geq0$, to each initial
value problem (\ref{V1eps}) and for all $t\in\left[0,\delta(s_0-s)\right)$ we
have
\[
\lim_{\eps\rightarrow0}\|u_{\eps}(t)-u_{0}(t)\|_{s}=0.
\]
\end{proposition}

We now turn back to our model. We wish to apply this theorem to the family of initial value problems \eqref{CPeps} and \eqref{CPVl}. The condition  \eqref{normest2} poses no problem because of Proposition~\ref{normestscale}. The only point remaining with direct application of  Proposition~\ref{Thconv} concerns  an analog of \eqref{estdif} (in Banach spaces  \eqref{newfamily}).
The latter problem is solved in the following proposition.

\begin{proposition}\label{prop_estdif}
Let \eqref{intcond} holds. Assume, additionally,
that
\begin{equation}\label{bddphi}
\bar{\phi}=\esssup_{x\in\X}\,\phi(x)<\infty.
\end{equation}
Let $C>C_{0}>0$ be fixed, and consider any $C^{\prime },C^{\prime \prime }$
such that $C_{0}\leq C^{\prime }<C^{\prime \prime }\leq C$. Then, for any $%
k\in \mathcal{K}_{C^{\prime }}$,
\begin{equation}\label{normdifest}
\bigl\Vert L_{\eps,\mathrm{ren}}^\triangle k-L^\triangle_V k\bigr\Vert_{\K_{C''}}\leq
\eps\biggl(\frac{M_1}{C''-C'}+\frac{M_2}{(C''-C')^2}\bigg)\Vert k\Vert_{\K_{C'}},
\end{equation}
where $M_1=\frac{m\beta C^2}{2e} \left\Vert \phi \right\Vert _{\infty }e^{C\beta }$, $M_2=\frac{4m C^2}{e^2} \left\Vert \phi \right\Vert _{\infty }e^{C\beta }$.
\end{proposition}
\begin{proof}
By \eqref{Ltriangleeps} and \eqref{LtriangleV}, one has
\begin{align*}
&\left( C^{\prime \prime }\right) ^{-\left\vert \eta \right\vert
}\bigl\vert \left( L_{\varepsilon ,\mathrm{ren}}^{\Delta }k\right) \left(
\eta \right) -\left( L_{V}^{\Delta }k\right) \left( \eta \right) \bigr\vert
 \\
\leq &\left( C^{\prime \prime }\right) ^{-\left\vert \eta \right\vert
}m\sum_{x\in \eta }\int_{\Gamma _{0}}\left\vert k\left( \eta \cup \xi \right)
\right\vert \\&\times\left\vert e^{-\varepsilon E^{\phi }\left( x,\eta \setminus
x\right) }e_{\lambda }\left( \frac{e^{-\varepsilon \phi \left( x-\cdot
\right) }-1}{\varepsilon },\xi \right) -e_{\lambda }\left( -\phi \left(
x-\cdot \right) ,\xi \right) \right\vert d\lambda \left( \xi \right).
\end{align*}%
Taking into account \eqref{obvest}, we obtain
\begin{align*}
&\left\vert e^{-\varepsilon E^{\phi }\left( x,\eta \setminus x\right)
}e_{\lambda }\left( \frac{e^{-\varepsilon \phi \left( x-\cdot \right) }-1}{%
\varepsilon },\xi \right) -e_{\lambda }\left( -\phi \left( x-\cdot \right)
,\xi \right) \right\vert \\
\leq &\left\vert e_{\lambda }\left( \frac{e^{-\varepsilon \phi \left(
x-\cdot \right) }-1}{\varepsilon },\xi \right) -e_{\lambda }\left( -\phi
\left( x-\cdot \right) ,\xi \right) \right\vert \\&+\left\vert \left(
1-e^{-\varepsilon E^{\phi }\left( x,\eta \setminus x\right) }\right)
e_{\lambda }\left( \frac{e^{-\varepsilon \phi \left( x-\cdot \right) }-1}{%
\varepsilon },\xi \right) \right\vert \\
=&\,e_{\lambda }\left( \phi \left( x-\cdot \right) ,\xi \right) -e_{\lambda
}\left( \frac{1-e^{-\varepsilon \phi \left( x-\cdot \right) }}{\varepsilon }%
,\xi \right) \\&+\left( 1-e^{-\varepsilon E^{\phi }\left( x,\eta \setminus
x\right) }\right) e_{\lambda }\left( \frac{1-e^{-\varepsilon \phi \left(
x-\cdot \right) }}{\varepsilon },\xi \right).
\end{align*}%
It can be easily seen by induction that for any $b_{i}\geq a_{i}>0$,
$1\leq i\leq n$, $n\in\N$,
\[
\prod_{i=1}^nb_{i}-\prod_{i=1}^na_{i}\leq \sum_{i=1}^n\left( b_{i}-a_{i}\right)
\prod_{\substack{j=1\\j\neq i}}^nb_{j}.
\]%
Then, for any $\eta\in\Ga_0$, $x\in\eta$, $\xi\in\Ga_0$, $\xi\cap\eta=\emptyset$, one has
\begin{align}\notag
&\left( 1-e^{-\varepsilon E^{\phi }\left( x,\eta \setminus x\right) }\right)
e_{\lambda }\left( \frac{1-e^{-\varepsilon \phi \left( x-\cdot \right) }}{%
\varepsilon },\xi \right) \\
%\leq &\sum_{y\in \eta \setminus x}\left(
%1-e^{-\varepsilon \phi \left( x-y\right) }\right) e^{-\varepsilon E^{\phi
%}\left( x,\left( \eta \setminus x\right) \setminus y\right) }e_{\lambda
%}\left( \frac{1-e^{-\varepsilon \phi \left( x-\cdot \right) }}{\varepsilon }%
%,\xi \right) \notag \\
\leq &\sum_{y\in \eta \setminus x}\varepsilon \phi \left( x-y\right)
e_{\lambda }\left( \phi \left( x-\cdot \right) ,\xi \right)\label{est1}
\end{align}%
and
\begin{align}\notag
0&\leq e_{\lambda }\left( \phi \left( x-\cdot \right) ,\xi \right) -e_{\lambda
}\left( \frac{1-e^{-\varepsilon \phi \left( x-\cdot \right) }}{\varepsilon }%
,\xi \right)\\ &\leq \sum_{y\in \xi }\left( \phi \left( x-y\right) -\frac{%
1-e^{-\varepsilon \phi \left( x-y\right) }}{\varepsilon }\right) e_{\lambda
}\left( \phi \left( x-\cdot \right) ,\xi \setminus y\right). \label{est2}
\end{align}
To estimate the right hand side of \eqref{est2}, we rewrite
\begin{align}\notag
&\quad\phi \left( x-y\right) -\frac{1-e^{-\varepsilon \phi \left( x-y\right) }}{\varepsilon }\\&=\frac{1}{\varepsilon ^{2}\phi ^{2}\left( x-y\right) }\left(
e^{-\varepsilon \phi \left( x-y\right) }+\varepsilon \phi \left( x-y\right)
-1\right) \varepsilon \phi ^{2}\left( x-y\right)\label{rewrite}
\end{align}
and consider the function
\[
f\left( t\right) =\frac{e^{-t}+t-1}{t^{2}}.
\]
Then, by an implicit differentiation, one has $f'(t)=-t^{-3}g(t)$,
where $g(t)=t+2e^{-t}+te^{-t}-2$. Next, $g'(t)=1-te^{-t}-e^{-t}$ and
$g''(t)=te^{-t}>0$ for $t>0$. Since $g'(0)=0$ we have $g'(t)>g'(0)=0$,
$t>0$ and then, since $g(0)=0$ one has $g(t)>0$, $t>0$. Therefore,
$f$ decays for $t>0$. Note also that $f(t)\to0$ as $t\to\infty$. As a result,
\begin{equation}\label{bdd}
0< f(t)< \lim_{s\to0}f(s)=\frac{1}{2},\quad t>0.
\end{equation}

Hence, by \eqref{est1}--\eqref{bdd}, one get
\begin{align*}
&\left( C^{\prime \prime }\right) ^{-\left\vert \eta \right\vert
}\bigl\vert \left( L_{\varepsilon ,\mathrm{ren}}^{\Delta }k\right) \left(
\eta \right) -\left( L_{V}^{\Delta }k\right) \left( \eta \right) \bigr\vert
\\
\leq\,& \frac{\varepsilon m}{2}\left( C^{\prime \prime }\right) ^{-\left\vert \eta
\right\vert }\sum_{x\in \eta }\int_{\Gamma _{0}}\left( C^{\prime }\right)
^{\left\vert \eta \right\vert +\left\vert \xi \right\vert }\left\Vert
k\right\Vert _{\mathcal{K}_{C^{\prime }}}\sum_{y\in \xi }\phi ^{2}\left(
x-y\right) e_{\lambda }\left( \phi \left( x-\cdot \right) ,\xi \setminus
y\right) d\lambda \left( \xi \right) \\
&+\varepsilon m \left( C^{\prime \prime }\right) ^{-\left\vert \eta
\right\vert }\sum_{x\in \eta }\int_{\Gamma _{0}}\left( C^{\prime }\right)
^{\left\vert \eta \right\vert +\left\vert \xi \right\vert }\left\Vert
k\right\Vert _{\mathcal{K}_{C^{\prime }}}\sum_{y\in \eta \setminus x}\phi
\left( x-y\right) e_{\lambda }\left( \phi \left( x-\cdot \right) ,\xi
\right) d\lambda \left( \xi \right) \\
=&\, \frac{\varepsilon m}{2}\left( \frac{C^{\prime }}{C^{\prime \prime }}\right)
^{\left\vert \eta \right\vert }C^{\prime }\left\Vert k\right\Vert _{\mathcal{%
K}_{C^{\prime }}}\sum_{x\in \eta }\int_{\Gamma _{0}}\int_{\mathbb{R}%
^{d}}\left( C^{\prime }\right) ^{\left\vert \xi \right\vert }\phi ^{2}\left(
x-y\right) dye_{\lambda }\left( \phi \left( x-\cdot \right) ,\xi \right)
d\lambda \left( \xi \right) \\
&+\varepsilon m \left( \frac{C^{\prime }}{C^{\prime \prime }}\right)
^{\left\vert \eta \right\vert }\left\Vert k\right\Vert _{\mathcal{K}%
_{C^{\prime }}}\sum_{x\in \eta }\sum_{y\in \eta \setminus x}\phi \left(
x-y\right) \int_{\Gamma _{0}}\left( C^{\prime }\right) ^{\left\vert \xi
\right\vert }e_{\lambda }\left( \phi \left( x-\cdot \right) ,\xi \right)
d\lambda \left( \xi \right) \\
\leq &\, \frac{\varepsilon m}{2}\left( \frac{C^{\prime }}{C^{\prime \prime }}\right)
^{\left\vert \eta \right\vert }C^{\prime }\left\Vert k\right\Vert _{\mathcal{%
K}_{C^{\prime }}}\beta \left\Vert \phi \right\Vert _{\infty }e^{C^{\prime
}\beta }\left\vert \eta \right\vert \\
&+\varepsilon m \left( \frac{C^{\prime }}{C^{\prime \prime }}\right)
^{\left\vert \eta \right\vert }\left\Vert k\right\Vert _{\mathcal{K}%
_{C^{\prime }}}\left\Vert \phi \right\Vert _{\infty }\left\vert \eta
\right\vert^2 e^{C^{\prime }\beta
}.
\end{align*}
For any $a\in(0,1)$, one has, cf. \eqref{max1},
\begin{equation}\label{max2}
\max_{t\geq 0}\left( t^2a^{t}\right) =\frac{4}{e^2\ln^2 a }.
\end{equation}
Therefore,
\begin{align*}
\esssup_{\eta\in\Ga_0}\,\left\vert \eta \right\vert^2\left( \frac{C^{\prime }}{C^{\prime \prime }}\right)
^{\left\vert \eta \right\vert }\leq\frac{4}{e^2(\ln C'-\ln C'')^2 }\leq \frac{4C^2}{e^2(C'-C'')^2 }.
\end{align*}

The rest of the proof is clear now.
\end{proof}

We are now in a position to prove the main result of this section.
\begin{theorem}
Let \eqref{intcond} and \eqref{bddphi} hold. Let $C>C_0>0$ be fixed
and $T_1=T_1(C_0,C)$ be given by~\eqref{T1}. Suppose also that $\bigl\{k_{0,\eps}, k_{0,V}\bigr\}_{\eps>0}\subset \K_{C_0}$ and, moreover,
\begin{equation}\label{initconv}
\lim_{\eps\rightarrow0}\|k_{0,\eps}-k_{0,V}\|_{\K_{C_0}}=0.
\end{equation}
Then, the equations \eqref{CPeps}, \eqref{CPVl} have solutions in
$\K_C$ and
\begin{equation}\label{timeconv}
\lim_{\eps\rightarrow0}\|k_{t,\eps}-k_{t,V}\|_{\K_C}=0.
\end{equation}
Moreover, if
\begin{equation}\label{initexp}
k_{0,V}(\eta)=e_\la(u_0,\eta), \quad \eta\in\Ga_0,
\end{equation}
for some $u_0\in L^{\infty}(\R^{d})$ such that
$0\leq u_0(x)\leq C_0$, a.a. $x\in\X$, then
\begin{equation}\label{timeexp}
k_{t,V}(\eta)=e_\la(u_t,\eta), \quad \eta\in\Ga_0,
\end{equation}
provided that $u_t$ is a solution to the non-linear
non-local mesoscopic equation
\begin{equation}\label{CP}
\begin{cases}
\dfrac{\partial }{\partial t}u_{t}\left( x\right) =-mu_{t}\left( x\right)
e^{-\left( u_{t}\ast \phi \right) \left( x\right) }+\uplambda \\[2mm]
u_{t}\left( x\right) \bigr\vert _{t=0}=u_0 \left( x\right),
\end{cases}
\end{equation}
with $u_t\in L^{\infty}(\R^{d})$ such that $0\leq u_t(x)\leq C$, for a.a. $x\in\X$ on the time interval $[0,T_1)$.
\end{theorem}
\begin{proof}
The existence of solutions to \eqref{CPeps} and \eqref{CPVl} in $\K_C$ was shown
before. The convergence \eqref{timeconv} may be established if we
apply Theorem~\ref{Thconv} to our model for the family~\eqref{newfamily} by using
Proposition~\ref{prop_estdif}. Therefore, we only need to prove the second
part of the statement. By uniqueness of solution to \eqref{CPVl}
it suffices to show that $k_{t,V}$, given by \eqref{timeexp}, solves
\eqref{CPVl} provided that $u_t$ solves \eqref{CP}. From \eqref{LPexp} it follows that
\begin{equation}
\frac{\partial}{\partial t}e_\la(u_t,\eta)=\sum_{x\in\eta}
\frac{\partial}{\partial t}u_t(x)e_\la(u_t,\eta\setminus x)\label{100ta}
\end{equation}
and, by \eqref{LPexp}, \eqref{intexp}, and \eqref{LtriangleV}, one has
\begin{align*}
\bigl(L_V^\triangle e_\la(u_t)\bigr)(\eta)=&-m\sum_{x\in \eta}e_\la(u_t,\eta)\int_{\Gamma _{0}}e_\la(u_t,\xi)e_{\lambda }\left( -\phi
\left( x-\cdot \right) ,\xi \right) d\lambda \left( \xi \right) \\  \label{502ta}
&+\uplambda \sum_{x\in
\eta }e_\la(u_t,\eta\setminus x)\\
=&-m\sum_{x\in
\eta }e_\la(u_t,\eta\setminus x)u_t(x)\exp\{-(u_t\ast \phi)(x)\}\\&+\uplambda \sum_{x\in
\eta }e_\la(u_t,\eta\setminus x).
\end{align*}
Comparing the right hand sides of the latter expression with \eqref{100ta} and taking
into account \eqref{CP} we conclude that $e_\la(u_t)$ solves \eqref{CPVl}.
This proves the theorem.
\end{proof}

\section{Solution to the mesoscopic equation}

In this section we will look more closely at the properties of solutions to the mesoscopic
equation \eqref{CP}. We will be interested in solutions which are bounded and continuous in space variable. Namely, for the time interval that is either $I=[0,a]$, $a>0$ or $I=\R_+:=[0,+\infty)$, a function $u_t(x)=u(t,x)$ is defined to be a solution to \eqref{CP} on $I$ iff $u\in C^1\bigl(I\to C_b(\X)\bigr)$ and it solves \eqref{CP}  for any $t\in I$. Here $C_b(\X)$ is the Banach space of bounded continuous functions on $\X$ with sup-norm denoted by $\|\cdot\|_\infty$.
%Note also, that, for $t=0$, we will consider the right derivative only.
For simplicity of notations, we continue to write $\dot{u}_t(x)$ for $\frac{\partial}{\partial t} u(t,x)$. It will cause no confusion if we use the same letter for the right derivative considered in \eqref{CP} for $t=0$.

We look as usual at an integral version of equation \eqref{CP}
\begin{equation}\label{intCP}
 u_t(x)=u_0(x)+\uplambda t - m \int_0^t u_\tau(x) e^{- (\phi\ast u_\tau)(x)} \,d\tau
\end{equation}
in the space $C\bigl(I\to C_b(\X)\bigr)$. It is easy to see that each solution to \eqref{intCP} will be continuously differentiable in $t$ in the sense of norm in $C_b(\X)$, and, hence, will be a solution to \eqref{CP}.
%It worth be noting that we do not need here Banach structure neither on the set $C^1\bigl(I\to C_b(\X)\bigr)$ nor on the set $C\bigl(I\to C_b(\X)\bigr)$.

Through this section we always suppose that $\uplambda>0$, $m>0$ and \eqref{intcond} holds.

\subsection{Existence, uniqueness, stability, and boundedness}

We begin with a general result on existence and uniqueness.
\begin{proposition}\label{exuniq}
Let $u_0\in C_b(\X)$. Then there exists $a>0$
such that the equation \eqref{CP} has a unique solution $u\in C^1\bigl([0,a]\to C_b(\X)\bigr)$.
If, additionally, $u_0\geq0$, then this solution may be extended to a nonnegative
solution $0\leq u\in C^1\bigl([0,+\infty)\to C_b(\X)\bigr)$.
\end{proposition}

\begin{proof}
Let us consider, for arbitrary $a>0$, $b>0$, the rectangle
\[
R:=\bigl\{(t,u):|t|\leq a, \|u-u_0\|_\infty\leq b,\, u\in C_b(\X) \bigr\}.
\]
Set
$
f(u):=-mue^{-u*\phi}+\uplambda.
$
Since $\|\phi*u\|_\infty\leq \beta \|u\|_\infty$ it follows for $(t,u)\in
R$ that
\begin{align}
\|f(u)\|_\infty &\leq \uplambda+m\|u\|_\infty \exp\bigl\{
\beta\|u\|_\infty\bigr\}\notag\\
& \leq \uplambda +m \bigl(\|u_0\|_\infty+b\bigr)\exp\bigl\{\beta \bigl(\|u_0\|_\infty+b\bigr)\bigr\}=:M.
\end{align}
Next, for any $(t,u_1)\in R, (t,u_2)\in R$, one has
\begin{align}
\|f(u_1)-f(u_2)\|_\infty&\leq m\|u_1 e^{-\phi*u_1}-u_2 e^{-\phi*u_1}\|_\infty
+m\|u_2 e^{-\phi*u_1}-u_2 e^{-\phi*u_2}\|_\infty\notag\\&\leq m\exp\bigl\{\beta\|u_1\|_\infty\bigr\}
\|u_1-u_2\|_\infty\notag\\&\qquad+m\|u_2\|_\infty\|e^{-\phi*u_2}\|_\infty\|e^{-\phi*(u_1-u_2)}-1\|_\infty.\label{est:1}
\end{align}
We denote, for fixed $|t|\leq a$, $x\in\X$,
\[
h_{t,x}(s)=\exp\bigl\{-s\bigl(\phi*(u_1(t,\cdot)-u_2(t,\cdot))\bigr)(x)\bigr\},
\quad s\in[0,1],
\] then
\begin{align*}
&|h_{t,x}(1)-h_{t,x}(0)|\leq 1\cdot\sup_{s\in[0,1]}|h'_{t,x}(s)|\\
\leq&\,\Bigl|\bigl(\phi*(u_1(t,\cdot)-u_2(t,\cdot))\bigr)(x)\Bigr|
\exp\Bigl\{\Bigl|\bigl(\phi*(u_1(t,\cdot)-u_2(t,\cdot))\bigr)(x)\Bigr|\Bigr\}\\
\leq& \beta \|u_1-u_2\|_\infty\exp\bigl\{\beta \|u_1-u_2\|_\infty\bigr\}.
\end{align*}
and we may now proceed to conclude from \eqref{est:1} that
\begin{align*}
&\quad\|f(u_1)-f(u_2)\|_\infty\\&\leq m\exp\bigl\{\beta(b+\|u_0\|_\infty)\bigr\}
\|u_1-u_2\|_\infty\\
&\qquad\quad+m(b+\|u_0\|_\infty) \exp\bigl\{\beta(b+\|u_0\|_\infty)\bigr\}
\beta \|u_1-u_2\|_\infty\exp\bigl\{\beta \|u_1-u_2\|_\infty\bigr\}\\
&\leq m\exp\bigl\{\beta(b+\|u_0\|_\infty)\bigr\}
 \bigl( 1 + \beta(b+\|u_0\|_\infty)e^{2\beta b} \bigr)\|u_1-u_2\|_\infty.
\end{align*}
Therefore, $f$ is a locally Lipschizian function. Thus, by e.g. \cite[Theorem~5.1.1]{LL1972},
for $a=\frac{b}{M}$,
there exists a unique strongly continuously differentiable function $u:[-a,a]\to C_b(\X)$ which satisfies \eqref{CP}.

To prove the existence of a global nonnegative solution to \eqref{CP} (on the whole $\R_+$), let us consider
the following equation
\begin{equation} \label{DS2}
\dot{u} = g(u),\qquad u(0):=u_0, \qquad g(u):= \uplambda - m u e^{- \phi
\ast|u|}.
\end{equation}
Then, clearly,
\[
\|g(u)\|_\infty \leq \uplambda + m\|u\|_\infty\bigl\| e^{- \phi
\ast|u|}\bigr\|_\infty\leq\uplambda + m\|u\|_\infty.
\]
We are now in a position to apply \cite[Theorem 5.6.1]{LL1972}. Namely, we set $h(r):=\uplambda+mr$ (which is an increasing function in $r$) and consider the following equation
\[
\dot{r}=h(r), \quad r(0):=\|u_0\|_\infty.
\]
Since the latter equation has a unique solution on the whole $\R_+$,
then the largest interval of existence of a solution to \eqref{DS2}
is also $\R_+$.

Next, suppose that $u_0(x)\geq0$, for all $x\in\X$.
Let us prove that the solution $u(t,x)$ to \eqref{DS2} is
also nonnegative, for any $t\geq0$, $x\in\X$. Suppose that there
is a point $x_0\in\X$ such that $u(\cdot,x_0)$ takes negative values for some $t$. Then there exists
$t_0:=\inf\{t>0\mid u(t,x_0)<0\}\geq 0$. Due to continuity of $u(\cdot,x_0)$, one has that
$u(t_0,x_0)=0$. Then, by \eqref{DS2},
\begin{equation}\label{poee}
 \dot{u}(t_0,x_0)=\uplambda>0.
\end{equation}
Let now $\{t_n\}_{n\in\N}\subset\{t>0\mid u(t,x_0)<0\}$ be such that $t_n\downarrow t_0$.
Then
\[
\dot{u}(t_0,x_0)=\lim_{n\to\infty}\frac{u(t_n,x_0)-u(t_0,x_0)}{t_n-t_0}=\lim_{n\to\infty}\frac{u(t_n,x_0)}{t_n-t_0}\leq 0,
\]
that is in contradiction to \eqref{poee}.
Hence, \eqref{DS2} has a global
nonnegative solution provided that $u_0$ is nonnegative. But, clearly,
this solution solves \eqref{CP} with the same initial condition.

It is worth noting that the arguments above about the positivity preservation of a solution may be also applied directly to the equation \eqref{CP}.
\end{proof}

In the sequel we will be concerned with properties of a non-negative solution
to \eqref{CP} on $[0,+\infty)$.
Consider the stationary space-homogeneous equation for \eqref{CP}. Namely,
\begin{equation}\label{stateqn}
 \uplambda =mue^{-\beta u}.
\end{equation}
To deal with this and subsequent equations it is necessary to mention the properties of the following function
\begin{equation}\label{dopfunc}
p\left( r\right) =re^{-r}, \quad r\geq0.
\end{equation}
It is immediate that
$0\leq p(r)\leq p(1)=\frac{1}{e}$ and, for any $\uplambda \in(0,\frac{1}{e})$, the equation $p(r)=\uplambda$ has two solutions $0<r_1<1<r_2$. The equation $\frac{\uplambda\beta}{m}=\beta u e^{-\beta u}$ (which is equivalent to \eqref{stateqn}) has also two solutions, say $\kappa_1$ and $\kappa_2$, such that $0<\beta \kappa_1<1<\beta \kappa_2$ provided $\frac{\uplambda\beta}{m}<\frac{1}{e}$.
Summarizing, we conclude that under condition
\begin{equation}\label{smallparam-s}
\uplambda <\frac{m}{\beta e},
\end{equation}
the equation \eqref{CP} has two positive equilibrium solutions $u(t,x) \equiv \kappa_{1}$ and $u(t,x) \equiv \kappa_{2}$ such that
\begin{equation}\label{root-est}
0<\kappa_1 < \frac{1}{\beta} < \kappa_2.
\end{equation}
The properties of equilibrium solutions are established by our next proposition.
\begin{proposition}
Let condition \eqref{smallparam-s} holds.
The equilibrium solution $u^\ast (t,x) \equiv
\kappa_1\in\bigl(0,\frac{1}{\beta}\bigr)$ is uniformly stable in the sense of Lyapunov, i.e., for any
$\varepsilon >0 $, there exists $\delta>0$ such that,
for any $t_1\geq0$, the inequality
\begin{equation}\label{delta}
\| u(t_1) - u^\ast \|_\infty < \delta
\end{equation}
implies
\begin{equation}
\| u(t) - u^\ast \|_\infty < \varepsilon, \quad t\geq t_1.
\end{equation}
Moreover, $u^*$ is asymptotically stable in the sense of Lyapunov,
i.e. if the inequality \eqref{delta} holds for some $t_1\geq0$ and
$\delta>0$, then
\begin{equation}
\lim_{t\to\infty} \|u(t)-u^*\|_\infty =0.
\end{equation}
\end{proposition}
\begin{proof}
It is well known, see e.g. \cite[Chapter VII]{DK1974}, that the statement will be proved by showing that the spectrum of the operator
$$
f'(u^*)v:=\frac{d}{ds}f(u^*+sv)|_{s=0},\quad v\in C_b(\X)
$$
belongs to the interior of the left half plane. Here, as in the proof of Proposition \ref{exuniq}, $
f(u)=-mue^{-u*\phi}+\uplambda.$ By definition, for
any $u,v\in C_b(\X)$,
\begin{align*}
f'(u)v&= -\Bigl( m v e^{- \phi
\ast (u+sv)}-m(u+sv)e^{- \phi
\ast (u+sv)}(\phi*v)\Bigr)\Bigr|_{s=0}\\&=-mve^{-\phi*u}+mue^{-\phi*u}(\phi*v).
\end{align*}
Taking into account that $u\equiv\kappa_1$ solves \eqref{stateqn}, we can assert that
\begin{align*}
\bigl(f'(\kappa_1)v\bigr)(x)&=-me^{-\beta \kappa_1}v(x)+m\kappa_1e^{-\beta
\kappa_1}(\phi*v)(x)\\&=\uplambda\Bigl((\phi*v)(x)-\frac{1}{\kappa_1}v(x)\Bigr)\\
&=\uplambda \int_\X \phi(x-y)\bigl( v(y)-v(x)\bigr)\,dx+\uplambda \Bigl(\beta-\frac{1}{\kappa_1}\Bigr)
v(x).
\end{align*}
%Since $\|\phi\ast v\|_\infty\leq \beta \|v\|_\infty$ and one has the equality for any constant $v$,
It is a simple matter to check that the spectrum of jump generator $(Av)(x):=\uplambda \int_\X \phi(x-y)\bigl( v(y)-v(x)\bigr)\,dy$ belongs to the circle $\{z\in\mathbb{C}: |z+\uplambda\beta|\leq \uplambda\beta\}$. The last claim is due to the fact that $\|\phi\ast v\|_\infty\leq \beta \|v\|_\infty$. From \eqref{root-est} we have $\beta<\frac{1}{\kappa_1}$ and, consequently,
the spectrum of $f'(\kappa_1)$ belongs to the interior of the left half plane which proves the statement. This finishes the proof.
\end{proof}

The reminder of this subsection will be devoted to the following refinement of Proposition \ref{exuniq}.

\begin{theorem}\label{Thbdd}
Let (cf. \eqref{smallparam-s})
\begin{equation}\label{smallparam}
 \uplambda \leq\frac{m}{\beta e}
\end{equation}
and let $\kappa_{1},\,\kappa_{2}$ be constant solutions to \eqref{stateqn}, in particular, $\kappa_1=\kappa_2$ if  \eqref{smallparam} is an equality.
Suppose that $0\leq u_0\in C_b(\X)$, with
$\Vert u_0\Vert_\infty\leq \kappa_2$. Then the equation \eqref{CP} has
a unique solution $0\leq u_t\in C_b(\X)$, $t\geq0$, such that
$\Vert u_t\Vert_\infty\leq \kappa_2$, for all $t\geq0$.

Moreover, for an arbitrary $c\in[\kappa_1, \kappa_2]$, the condition $\kappa_1\leq u_0(x)\leq c$, $x\in\X$, yields $\kappa_1\leq u_t(x)\leq c$, $x\in\X$.
\end{theorem}
\begin{proof}
Consider, for a fixed $T>0$, the Banach
space
\[
X_{T}=C\left( \left[ 0,T\right] \rightarrow C_b(\mathbb{R}^{d})\right)
\]%
with norm%
\[
\left\Vert w\right\Vert _{T}=\max_{t\in \left[ 0,T\right] }\left\Vert
w_{t}\right\Vert _{\infty }.
\]
For an arbitrary $c>0$, let $B_{c,T}^{+}$ be the set of all $w\in X_{T}$ such that $%
\left\Vert w\right\Vert _{T}\leq c$ and $w_{t}\left( x\right) \geq 0$ for
all $t\in \left[ 0,T\right] $ and for a.a. $x\in \mathbb{R}^{d}$. Clearly, $%
B_{c,T}^{+}$ with a metric induced by the norm $\left\Vert \cdot \right\Vert
_{T}$ constitutes a complete metric space.

For any $v\in $ $B_{c,T}^{+}$ and for any $0\leq u_0 \in C_b(\mathbb{R%
}^{d})$ with $\left\Vert u_0 \right\Vert _{\infty }\leq c$,
we consider a mapping $u=\Phi v$ which maps $v$ into the solution of the linear equation%
\[
\begin{cases}
\dfrac{\partial }{\partial t}u_{t}\left( x\right) =-mu_{t}\left( x\right)
e^{-\left( v_{t}\ast \phi \right) \left( x\right) }+\uplambda \\[2mm]
u_{t}\left( x\right) \bigr\vert _{t=0}=u_0 \left( x\right).
\end{cases}
\]%
Namely,%
\begin{align}
\notag\left( \Phi v\right) _{t}\left( x\right) =&\,\exp \left\{
-m\int_{0}^{t}e^{-\left( v_{s}\ast \phi \right) \left( x\right) }ds\right\}
u_0 \left( x\right) \\&+\uplambda \int_{0}^{t}\exp \left\{ -m\int_{\tau }^{t}e^{-\left(
v_{s}\ast \phi \right) \left( x\right) }ds\right\} d\tau.\label{Phi}
\end{align}
Since $0\leq v_{s}\left( x\right) \leq c$, $x\in\R^{d}$, and $\phi \geq 0$ it follows that
\[
-e^{-\left( v_{s}\ast \phi \right) \left( x\right) }\leq -e^{-\beta c}.
\]%
Therefore, taking into account that initial function $u_{0}$ satisfies $0\leq u_0 \left( x\right) \leq c$ we
get%
\begin{align}
0 &\leq \left( \Phi v\right) _{t}\left( x\right) \leq \exp \left\{
-me^{-\beta c}t\right\} c+\uplambda \int_{0}^{t}\exp \left\{ -me^{-\beta c}\left( t-\tau
\right) \right\} d\tau \notag\\
&=\exp \left\{ -me^{-\beta c}t\right\} c+\frac{\uplambda }{m} e^{\beta c}\left( 1-\exp \left\{
-me^{-\beta c}t\right\} \right) \notag\\
&\leq \max \left\{ c,\frac{\uplambda }{m} e^{\beta c}\right\}.\label{inside}
\end{align}%
By properties of the function \eqref{dopfunc} and inequality \eqref{smallparam} we conclude that
\[
\beta ce^{-\beta c}\geq \frac{\beta\uplambda }{m}, \quad c\in [\kappa_1,\kappa_2].
\]
Hence, for such $c$, $\max \left\{ c,\frac{\uplambda }{m} e^{\beta c}\right\}=c$, and \eqref{inside} shows that $\Phi :B_{c,T}^{+}\rightarrow B_{c,T}^{+}$.

Next, let us show that $\Phi $ is a contraction mapping on $B_{c,T}^{+}$. Let $v,w\in B_{c,T}^{+}$. Then, using an elementary inequality $\left\vert e^{-a}-e^{-b}\right\vert \leq \left\vert a-b\right\vert $ for $a,b\geq 0$, we obtain
\begin{eqnarray*}
&&\left\vert \left( \Phi v\right) _{t}\left( x\right) -\left( \Phi w\right)
_{t}\left( x\right) \right\vert \\
&\leq &\left\vert \exp \left\{ -m\int_{0}^{t}e^{-\left( v_{s}\ast \phi
\right) \left( x\right) }ds\right\} -\exp \left\{ -m\int_{0}^{t}e^{-\left(
w_{s}\ast \phi \right) \left( x\right) }ds\right\} \right\vert u_0 \left(
x\right) + \\
&&+\uplambda \int_{0}^{t}\left\vert \exp \left\{ -m\int_{\tau }^{t}e^{-\left(
v_{s}\ast \phi \right) \left( x\right) }ds\right\} -\exp \left\{ -m\int_{\tau
}^{t}e^{-\left( w_{s}\ast \phi \right) \left( x\right) }ds\right\}
\right\vert d\tau \\
&\leq &cm\left\vert \int_{0}^{t}e^{-\left( v_{s}\ast \phi \right) \left(
x\right) }ds-\int_{0}^{t}e^{-\left( w_{s}\ast \phi \right) \left( x\right)
}ds\right\vert \\&&+\uplambda m\int_{0}^{t}\left\vert \int_{\tau }^{t}e^{-\left( v_{s}\ast
\phi \right) \left( x\right) }ds-\int_{\tau }^{t}e^{-\left( w_{s}\ast \phi
\right) \left( x\right) }ds\right\vert d\tau \\
&\leq &cm\int_{0}^{t}\left\vert \left( v_{s}\ast \phi \right) \left( x\right)
-\left( w_{s}\ast \phi \right) \left( x\right) \right\vert
ds\\&&+\uplambda m\int_{0}^{t}\int_{\tau }^{t}\left\vert \left( v_{s}\ast \phi \right)
\left( x\right) -\left( w_{s}\ast \phi \right) \left( x\right) \right\vert
dsd\tau.
\end{eqnarray*}%
Next,
\[
\left\vert \left( v_{s}\ast \phi \right) \left( x\right) -\left( w_{s}\ast
\phi \right) \left( x\right) \right\vert \leq \int_{\mathbb{R}%
^{d}}\left\vert v_{s}\left( y\right) -w_{s}\left( y\right) \right\vert \phi
\left( x-y\right) dy\leq \beta \left\Vert v-w\right\Vert _{T}.
\]%
Therefore,%
\[
\left\vert \left( \Phi v\right) _{t}\left( x\right) -\left( \Phi w\right)
_{t}\left( x\right) \right\vert \leq c\beta mT\left\Vert v-w\right\Vert
_{T}+\uplambda \beta m \allowbreak \frac{T^{2}}{2}\left\Vert v-w\right\Vert _{T}.
\]%
As a result, for $T>0$ such that
\begin{equation}\label{timeint}
 \uplambda \beta m \allowbreak \frac{T^{2}}{2}+c\beta
mT<1,
\end{equation}
we obtain that there exists a unique fixed point of $\Phi
:B_{c,T}^{+}\rightarrow B_{c,T}^{+}$ that provides a unique nonnegative
solution to \eqref{CP} bounded by $c>0$ on the interval $\left[ 0,T%
\right] $.
Repeated application of the above proof to the initial function $0\leq u_{T}\left( x\right) \leq c$ enables us to extend solution to the time interval $\left[ T,2T\right] $ and hence to $\R_{+}$.

These arguments cover, clearly, the first statement of the theorem, when $c=\kappa_2$. What is left is to show that $u_0\geq\kappa_1$ implies $u_t\geq\kappa_1$.
%Of course, the set $B^+_{\kappa_1,\kappa_2,T}$ of all functions from $w\in X_T$, such that $\kappa_1\leq w_t(x)\leq \kappa_2$, %$x\in\X$, $t\in[0,T]$, is a metric space too.
Let us show that $u_0(x)\geq\kappa_1$, $v_t(x)\geq\kappa_1$ yields $(\Phi v)_t(x)\geq\kappa_1$, where $\Phi$ is given by \eqref{Phi}. Since $\kappa_1$ solves \eqref{stateqn}, one has
\begin{align*}
\notag\left( \Phi v\right) _{t}\left( x\right) &\geq \exp \left\{
-m\int_{0}^{t}e^{-\kappa_1\beta}ds\right\}
\kappa_1 +\uplambda \int_{0}^{t}\exp \left\{ -m\int_{\tau }^{t}e^{-\kappa_1\beta}ds\right\} d\tau \\&=
\exp\bigl\{-me^{-\kappa_1\beta}t\bigr\}\kappa_1+\frac{\uplambda }{m}e^{\kappa_1\beta}\Bigl(1-\exp\bigl\{-me^{-\kappa_1\beta}t\bigr\}\Bigr) =\kappa_1.
\end{align*}
According to the proof of the fixed point theorem, the fixed point $u$ of $\Phi$ (which will be the solution to \eqref{CP}) may be obtained as a limit of $\Phi^n u_0$ in the space $B^+_{c,T}$. If we subsequently choose $v=u_0$, $v=\Phi u_0$, $v=\Phi^2 u_0$, $\ldots$ in our previous arguments, we can assert that $u_0\geq \kappa_1$ implies $\Phi^n u_0\geq\kappa_1$, hence that $u\geq\kappa_1$ as well.
This completes the proof.
\end{proof}

\begin{remark}[Comparison principle]
It is easily seen that the proof of Theorem~\ref{Thbdd} is mainly based on the fact that the mapping \eqref{Phi} is increasing. Namely, if $0\leq v^1_t(x)\leq v^2_t(x)\leq\kappa_2$, $t\in[0,T]$, $x\in\X$, then $0\leq (\Phi v^1)_t(x)\leq (\Phi v^2)_t(x)\leq\kappa_2$, $t\in[0,T]$, $x\in\X$. On the other hand, the mapping \eqref{Phi} evidently depends on $u_0$. As a matter of fact, $\Phi=\Phi_{u_0}$, and this dependence is monotone: the condition
\begin{equation}\label{monot-init}
 0\leq u_0^1(x)\leq u_0^2(x)\leq\kappa_2, \quad x\in\X,
\end{equation}
implies $0\leq (\Phi_{u_0^1} v)_t(x)\leq (\Phi_{u_0^2} v)_t(x)$, $t\in[0,T]$, $x\in\X$. Combining both monotonicity, we deduce that for any $n\in\N$, the initial condition \eqref{monot-init} yields $0\leq (\Phi_{u_0^1}^n u_0^1)_t(x)\leq (\Phi_{u_0^2}^n u_0^2)_t(x)\leq\kappa_2$, $t\in[0,T]$, $x\in\X$. Hence, by the proof of the fixed point theorem, one can pass here to the limit in $n$ and derive that
\begin{equation}\label{monot-time}
 0\leq u_t^1(x)\leq u_t^2(x)\leq\kappa_2, \quad x\in\X,\quad t\in[0,T],
\end{equation}
where $u_t^1$, $u_t^2$ are solutions to \eqref{CP} with the initial conditions $u_0^1$, $u_0^2$ satisfying \eqref{monot-init}. The extension of the property \eqref{monot-time} on the whole $[0,+\infty)$ is evident. Taking into account that $u_t(x)\equiv\kappa_1$ is a solution to \eqref{CP}, one can obviously replace lower bound $0$ by $\kappa_1$ simultaneously in inequalities \eqref{monot-init} and \eqref{monot-time}.
\end{remark}

\subsection{Aggregation properties: growth process}

%As we shown above, for uniformly small (in space) initial data in
%\eqref{CP} we obtain uniformly small (in space and time) solution
%to \eqref{CP}. Below we consider the properties of a solution to \eqref{CP} under the assumption that the initial data has a big enough growth
%in a proper domain.
In this subsection, we proceed with the study of properties of solution to \eqref{CP}. We will be concerned with the situation when initial function is large enough in some domain of $\R^{d}$, contrary to the situation described in
Theorem \ref{Thbdd}.

%For the technical purposes we recall once again the most important properties of solutions to the problem $p(b):=be^{-b}=R$ (cf. \eqref{dopfunc}). As it was mentioned in the previous subsection, the latter equation has two solutions
%$b_{1}\leq1\leq b_{2}$ provided $R\leq 1/e$.

Let $A\in {\B}_{\mathrm{b}} (\X)$ be arbitrary. We set
\begin{equation}\label{bsf}
\Ph_{A}:=\inf_{x\in A}s_{A}(x),\quad \text{where}\quad
s_{A}(x):=\int_{A} \phi(x-y) dy.
\end{equation}

\begin{remark}
Suppose that there exists $\varepsilon >0$ such that $\phi(x)>0$, $x\in[-\varepsilon,\varepsilon]^d$ and $A$ contains at least one interior point. Then it is evident that
$\Ph_{A}>0$.
\end{remark}
\begin{remark}
Condition $\phi\in L^{1}(\R^{d})$ makes it obvious that $s_{A}(x)\rightarrow 0$ as $x\rightarrow \infty$.
\end{remark}

%For arbitrary $\theta\geq 0$ we set
%\begin{align}
%S_{A,\theta} \quad :=\quad & \quad[1/\Ph_{A}, \infty), \quad  \text{if} \quad \frac{\la\Ph_{A}}{me^{\theta}}>\frac{1}{e}\\
%& \quad [b_{\theta}, \infty), \quad \text{otherwise},
%\end{align}
%where $b_{\theta}$ is a larger solution to the following equation
%\begin{equation}
%p(\Ph_{A}b_{\theta})=\frac{\la\Ph_{A}}{me^{\theta}}.\label{dopeq}
%\end{equation}
Let us consider an equation similar to \eqref{stateqn}
\begin{equation}
b\,e^{-\Ph_{A}b}=\frac{\uplambda}{m}\theta,\label{dopeq}
\end{equation}
where $\theta\leq 1$ is an arbitrary constant. Let $\hat{b}=b(1/4)$ be the larger solution to the equation \eqref{dopeq} for $\theta=1/4$ provided $\frac{\uplambda\Ph_{A}}{4m}\leq 1/e$ and let $\hat{b}=1/\Ph_{A}$ otherwise.
\begin{theorem}\label{l-th1}
Let $A\in{\B}_{\mathrm{b}} (\X)$ be arbitrary. Suppose that $\Ph_{A}>0$ and $0\leq u_0\in C_b(\X)$ is  an initial condition to \eqref{CP} such that
\begin{equation}\label{u0m}
b < u(x,0) < \kappa b , \quad x\in A,
\end{equation}
for a pair $(b,\kappa)$ satisfying $b\geq \hat{b}$, $\kappa >1$, and
\begin{equation}
b \, e^{-\Ph_{A}b}<\frac{\uplambda}{m}\frac{1}{\kappa}\left(1-\frac{1}{\kappa}\right).\label{Bb}
\end{equation}
Then, the corresponding solution $u(x,t)$ to the equation \eqref{CP} grows on $A$ to infinity in such a way that
$$\uplambda/\kappa<v(b,\,\kappa)\leq \dot{u}(x,t)
\leq \uplambda,$$
where
 $v(b,\,\kappa):=\uplambda-{\kappa}mbe^{-\Ph_{A}b}$.
Consequently,
\begin{equation}\label{Gr}
b+\frac{\uplambda}{\kappa} t<b + v(b,\,\kappa)t \leq u(x,t) \leq \kappa b+\uplambda t =\kappa(b+\frac{\uplambda}{\kappa} t), \; x\in A, \; t \ge 0.
\end{equation}
If, additionally, $u(x,0) < \kappa b$ and $s_{A}(x)>0$ for some $x\in \R^{d}\setminus A$, then there exists a time $t(x)>0$ such that solution to the equation \eqref{CP}
monotonically increases at $x$ to infinity with the speed
\begin{equation}\label{sdsadas}
\dot{u}(x,t)\geq v(b,\,\kappa),\quad t \ge t(x).
\end{equation}
\end{theorem}
\begin{remark}
By the definition of $\hat{b}$ it is obvious that
$$
b\,e^{-\Ph_{A}b}<\frac{\uplambda}{4m}, \quad b\geq \hat{b}
$$
Thus, there exists $\kappa>1$ such that \eqref{Bb} holds.
\end{remark}
\begin{remark}
Let $\kappa>1$ be arbitrary and fixed. It follows easily that $v(b,\kappa)$ tends to $\uplambda$ as $b\rightarrow\infty$.
\end{remark}
\begin{remark}
Let $\kappa^*(b)>1$ be supremum over all $\kappa>1$ satisfying \eqref{Bb} for each fixed $b\geq \hat{b}$ . It is a simple matter to check that $\kappa^*(b)$ tends to $\infty$ as $b\rightarrow\infty$.
\end{remark}
\begin{proof}
We start with an evident representation
\begin{equation}\label{repres}
 u(x,t)=u(x,t_0)+\int_{t_0}^t \dot{u}(x,s)\,ds, \quad 0\leq t_0\leq t.
\end{equation}
By Proposition \ref{exuniq}, $u_0\geq0$ implies $u_t\geq0$, and hence,
\begin{equation}\label{estforder}
 \dot{u}_t(x)=\uplambda-mu_t(x)e^{-(\phi*u_t)(x)}\leq\uplambda, \quad t\geq0, \quad x\in\X.
\end{equation}
As a result, taking $t_0=0$ in \eqref{repres}, we immediately derive the
estimate from above in \eqref{Gr} for all $x\in\X$.

Our next concern will be the estimate from below. We see at once that one can establish bijective correspondence
$$\Theta:[\,\hat{b},\,\infty)\rightarrow \left(0, \min\left\{\frac14,\,\frac{m}{\uplambda e \Ph_{A} }\right\}\right]$$ between larger solution to \eqref{dopeq} and $\theta$. Namely, for each $b\geq \hat{b}$ there exists unique $\theta$ such that $b$ will be the larger solution to \eqref{dopeq} and
\begin{equation}
\Theta(b):=\theta.\label{thett}
\end{equation}
It is worth pointing out that $b>b'$ implies $\Theta({b'})>\Theta({b})$. Moreover, it is a simple matter to verify that
the estimate \eqref{Bb} is equivalent to
\begin{equation}
\Theta(b)<{\kappa}^{-1}\left(1-\kappa^{-1}\right). \label{Bb1}
\end{equation}
As a result, $\Theta(b)<\min\{{\kappa}^{-1},1/4\}$.  Another way of stating \eqref{Bb} for any fixed $\kappa>1$ is to say that $b\geq \hat{b}$ has to be larger than or equal to the minimal $b$ satisfying \eqref{Bb1}.
An easy computation shows that
\begin{equation}
v(b,\,\kappa):=\uplambda(1-\kappa\,\Theta(b)).\label{VV}
\end{equation}

Let us fix now an arbitrary $T>0$ and consider $t\in[0,T]$. Since $\dot{u}_t\in C\bigl([0,T]\to C_b(\X)\bigr)$ it follows that $\dot{u}_t$ is uniformly continuous on $[0,T]$. In particular, it means that for an arbitrary $\eps>0$, there exists $\delta>0$ such that, for any $t_0\in[0,T-\delta)$ and for any $t\in[t_0,t_0+\delta]$,
\begin{equation}\label{sasad}
 \|\dot{u}_t-\dot{u}_{t_0}\|_\infty<\eps.
\end{equation}
Let $\theta^*\in(\kappa,\,\Theta(b)^{-1}\left(1-\kappa^{-1}\right)]$ be arbitrary and fixed. Suppose that for some $t_0\in[0,T-\delta)$, one has the following bound
\begin{equation}\label{saddsa}
 \dot{u}(x,t_0)>v(b,\theta^*), \quad x\in A,
\end{equation}
where $v(b,\theta^*):=\uplambda-m \theta^* b e^{- \Ph_{A} b }=\uplambda\left(1-\theta^*\Theta(b)\right)$.
Then, for $\eps<v(b,\theta^*)$, we conclude from \eqref{sasad}, that
\begin{equation*}
 \dot{u}(x,t)>\dot{u}(x,t_0)-\eps\geq v(b,\theta^*)-\eps>0, \quad x\in A, \quad t\in[t_{0},t_{0}+\delta].
\end{equation*}
Thus, by \eqref{repres}, we see that the bound \eqref{saddsa} for the derivative at point $t_0$ yields the estimate
\begin{equation}\label{adsg}
 u(x,t)\geq u(x,t_0), \quad x\in A,\quad t\in[t_{0},t_{0}+\delta].
\end{equation}
For any $0\leq u\in C_b(\X)$, one has
\begin{align}
(\phi \ast u)(x) &= \int_\X u(y) \phi(x-y) dy \ge
 \int_A u(y) \phi(x-y) dy \notag\\&\ge
\min_{y \in A} u(y) \cdot \min_{x\in A} \int_{A} \phi(x-y) dy
= \min_{y\in A} u(y) \cdot \Ph_{A}, \quad x\in A. \label{estconv}
\end{align}
Therefore, \eqref{u0m} and $\kappa<{\theta^*}$ shows that for any $x \in A$,
\[
\phi \ast u_0(x) > b \Ph_{A},\qquad
m u_0
e^{- \phi \ast u_0} < m \kappa b e^{- \Ph_{A} b }<m {\theta^*} b e^{- \Ph_{A} b }.
\]
Consequently,
\begin{equation}\label{estdot0}
\dot{u}_0 = \uplambda - m u_0 e^{- \phi \ast u_0} > \uplambda-m \theta^* b e^{- \Ph_{A} b }=v(b,\theta^*)
 , \quad x\in A,
\end{equation}
that is \eqref{saddsa}, for $t_0=0$. Hence, by \eqref{adsg} and arguments above,
\begin{equation}\label{estontau}
 b <u(x, 0)\leq u(x, t) < \kappa b +\uplambda t, \quad t \in [0, \delta], \quad x\in A.
\end{equation}
Now we would like to get the inequality \eqref{estdot0} valid for small positive times. By \eqref{estconv} and \eqref{estontau} we have
\begin{equation}\label{estnewdop}
 m u_t (x) e^{-  \phi \ast u_t (x)} < m (\kappa b +
\uplambda t) e^{-  \Ph_{A} b}, \quad t \in [0, \delta], \quad x\in A.
\end{equation}
As a result, for all $t\in(0,\delta]$ satisfying
\begin{equation}\label{ontau}
  m \kappa b e^{-
\Ph_{A} b } + \uplambda m t e^{- \Ph_{A} b} \le
\uplambda - v(b,\theta^*)={\theta^*}mbe^{-\Ph_{A}b}
\end{equation}
we get
$$\dot{u}(x,t)=\uplambda -m u_t (x) e^{-  \phi \ast u_t (x)} >v(b,\theta^*).$$
Because of the strict inequality $\kappa <{\theta^*}$, there exists
\begin{equation}\label{sat1}
 \tau':= \frac{b\left({\theta^*}-\kappa\right) }{\uplambda}>0
\end{equation}
such that \eqref{ontau} holds for $t<\tau'$.
Then, for all $t\in(0,\min\{\tau',\delta\}]$
\begin{equation}\label{estderbel}
 \dot{u}(x,t) > v(b,\theta^*),  \quad x\in A,
\end{equation}
and hence, the representation \eqref{repres} yields the improved estimate for \eqref{estontau}
\begin{equation}\label{ytau}
b + v(b,\theta^*) t < u(x, t) < \kappa b +\uplambda t, \quad t\in [0,\min\{\tau',\delta\}], \quad x\in A.
\end{equation}
Set $\tau:=\min\{\tau',\delta\}$. Our next goal will be to show \eqref{ytau} for $t> \tau$. By \eqref{adsg}, \eqref{estderbel}, and \eqref{ytau} we conclude that \eqref{saddsa} holds for $t_0=\tau$ and
% Therefore, by  \eqref{adsg} and \eqref{ytau}
\begin{equation}\label{newest1}
 b+v(b,\theta^*)\tau< u(x,\tau)<u(x,t)<\kappa b+\uplambda t, \quad t\in[\tau,\tau+\delta], \quad  x\in A.
\end{equation}
To have \eqref{newest1} valid for $t\in[\tau,\tau+\delta]$,
we would like to obtain the estimate \eqref{estderbel} for $t>\tau$. First, we define $b_\tau := b + v(b,\theta^*) \tau$ and $\tau_1:=t-\tau$. Then for any $t\in[\tau,\tau+\delta]$ and  $x\in A$ we conclude from \eqref{estconv} and \eqref{newest1} that
\begin{equation}\label{sertgd}
 m u_{t} (x)e^{-\phi \ast u_t(x)} < m
(b_\tau + (\kappa -1)b +(\uplambda-v(b,\theta^*))\tau +\uplambda \tau_1) e^{-  \Ph_{A} b_\tau}
\end{equation}
$$
=m\left(b_\tau + (\kappa -1)b +v(b,\theta^*)\left(\frac{\theta^*\Theta(b)}{1 - \theta^* \Theta(b) }\right)\tau +\uplambda \tau_1\right) e^{-  \Ph_{A} b_\tau},
$$
which is clear from \eqref{VV}. Combining \eqref{Bb1} with $\theta^* \leq \Theta(b)^{-1} \left(1-\kappa^{-1}\right)$ we see that the latter expression will be less than or equal to
$$
m(b_\tau + (\kappa-1)b_{\tau}  +\uplambda \tau_1) e^{-  \Ph_{A} b_\tau}.
$$
Since $1<\kappa < \theta^* \leq \Theta(b)^{-1} \left(1-\kappa^{-1}\right)$ we have $\Theta(b)<1/\theta^*$. As a result,
$b_\tau>b>\Theta^{-1}(1/\theta^*)$, where $\Theta^{-1}$ is inverse transform to $\Theta$. Therefore
 $$be^{-\Ph_A b}>b_\tau e^{-\Ph_A b_\tau}$$ and, in consequence,
$$
m u_{t} (x)e^{-\phi \ast u_t(x)} < m \kappa be^{-  \Ph_{A} b} + \uplambda m \tau_1 e^{-  \Ph_{A} b}= \kappa \uplambda \Theta(b) + \uplambda m \tau_1 e^{-  \Ph_{A} b}.
$$
Thus, for all $\tau_{1}\leq \tau' = \frac{(\theta^*-\kappa)b}{\uplambda}$

$$m u_{t} (x)e^{-\phi \ast u_t(x)}<\kappa \uplambda \Theta(b) + \uplambda \tau_1 e^{-  \Ph_{A} b} < \uplambda \theta^{*} \Theta(b),$$
and hence
$$\dot{u}(x,t)=\uplambda -m u_t (x) e^{-  \phi \ast u_t (x)} >v(b,\theta^*).$$
Since $\tau\leq \tau'$ the inequality \eqref{estderbel} is also satisfied for $t\in[\tau,2\tau]$.

As was mentioned before, this gives the improved estimate \eqref{newest1}, namely \eqref{ytau} holds, for all $t\in[\tau,2\tau]$, $x\in A$. The same arguments show that \eqref{ytau} holds, for all $t\in[2\tau,3\tau]$ and so on. As a result, we obtain \eqref{ytau} on the whole $[0,T]$. Since $T>0$ and $\theta^*>\kappa$ are arbitrary, the first statement of the theorem is proved.

What is left is to show \eqref{sdsadas}. Let $x\in\R^{d}\setminus A$ be arbitrary and fixed. We consider the following cases: (a) $s_{A}(x) \geq \Ph_{A}$; (b) $s_{A}(x) < \Ph_{A}$.

(a) in this situation the proof of \eqref{sdsadas} is straightforward. Indeed, all bounds obtained in the proof of the first statement of the theorem up to \eqref{adsg} remain valid for each $x\in\R^{d}\setminus A$. Replacing  \eqref{estconv} by
\begin{align}
(\phi \ast u)(x) &= \int_\X u(y) \phi(x-y) dy \ge
 \int_A u(y) \phi(x-y) dy \notag\\&\ge
\min_{y \in A} u(y) \cdot s_{A}(x), \quad x\in \R^{d}\setminus A. \label{estconv1}
\end{align}
and using $s_{A}(x) \geq \Ph_{A}$ we see at once analogously to \eqref{estontau} that
\begin{equation}
u(x, t) < \kappa b +\uplambda t, \quad t \in [0, \delta], \quad x\in \R^{d}\setminus A,
\end{equation}
and all subsequent bounds are satisfied for $x\in\R^{d}\setminus A$, which is desired conclusion.

(b) we define $\hat{b}_x$ analogously to $\hat{b}$ for
$$
be^{-s_{A}(x)b}=\frac{\uplambda}{4m}\,
$$
and a bijective map $\Theta_{x}:[\hat{b}_x,\,\infty)\rightarrow \left(0, \min\left\{\frac14,\,\frac{m}{e \uplambda s_{A}(x) }\right\}\right]$ as an analogue to \eqref{thett}
\begin{equation}
be^{-s_{A}(x)b}=\frac{\uplambda}{m}\,{\Theta_{x}(b)}, \quad b\geq \hat{b}_x\label{dopeq11}.
\end{equation}
%In fact, the larger solution $b$ to \eqref{dopeq11} depends now on $x\in\R^{d}\setminus A$. To stress this fact we write $b(x)$ for such $b$.
Let $\Theta_{x}^{-1}$ denotes the inverse map. We define for any $(b,\kappa)$ satisfying \eqref{Bb}
\begin{equation}
t(x):=(b(x)-b)/v(b,\kappa),\label{timee}
\end{equation}
where  $b(x):=\Theta_{x}^{-1}(\Theta(b))> b$, which is clear from $s_{A}(x) < \Ph_{A}$.
%Let $b\geq b_{x}(1)\geq b(1)$. In this case we can proceed analogously to the proof of the first statement of the theorem.
%Let $b_{x}(1)\geq b \geq b(1)$. In fact, the larger solution $b$ to \eqref{dopeq11} depends now on $x\in\R^{d}\setminus A$. To stress this fact we write $b(x)$ for such $b$. Set $t(x):=(b(x)-b)/v(b,\kappa)$.
According to the arguments in the first part of the proof of this theorem we have
\begin{equation}\label{BB1}
u(y, t(x)) \geq b + v(b,\kappa) t(x) = b(x), \quad y \in A.
\end{equation}
Moreover, as it was mentioned before, \eqref{repres} and \eqref{estforder} imply that
\begin{equation}\label{BA1}
u(x, t(x)) \leq \kappa b + \uplambda t(x)  = b(x)+ (\kappa-1)b+(\uplambda-v(b,\kappa))t(x)
\end{equation}
It follows from estimates (\ref{estconv1}) and (\ref{BB1}) that
$ (\phi\ast u)(x, t(x))\geq s_{A}(x) b(x)$. Analogously to \eqref{sertgd} (with $\tau_1=0$,  $\tau$ replaced by $t(x)$, $\Ph_{A}$ replaced by $s_{A}(x)$ and $b_{\tau}$ replaced by $b(x)$) it may be concluded that
\begin{equation}\label{sad213}
 \dot{u}(x,t(x))\geq \uplambda-\kappa m b(x) e^{-s_{A}(x)b(x)}= \uplambda-\kappa m b e^{-\phi_{A}b}=v(b,\kappa),
\end{equation}
which is clear from the definition of $b(x)$.
Next, taking into account \eqref{Gr} and \eqref{BB1} we get
\begin{equation}\label{rrtt}
 u(y, t(x) + \tau) \geq u(y, t(x)) +v(b,\kappa) \tau \geq b(x)
+ v(b,\kappa) \tau, \; y \in A.
\end{equation}
Then, by \eqref{estconv1} and \eqref{rrtt},
\begin{equation}\label{sdaasdasd}
 (\phi \ast u)(x, t(x)+\tau) \geq s(x) (b(x)+ v(b,\kappa)
\tau).
\end{equation}
As a result,
\begin{align*}
m u(x, t(x)+\tau) e^{-\phi \ast u(x, t(x) + \tau)} &\leq m(\kappa b +   \uplambda(t(x)+\tau))
 e^{- s_{A}(x) (b(x)+ v(b,\kappa)
\tau)}
\end{align*}
$$
\leq m(b(x) +(\kappa -1)b + (\uplambda-v(b,\kappa))(t(x)+\tau) +v(b,\kappa)\tau)
 e^{- s_{A}(x) (b(x)+ v(b,\kappa)
\tau)}.
$$
Analogously to \eqref{sertgd} (with $\tau_1=0$,  $\tau$ replaced by $t(x)+\tau$, $\Ph_{A}$ replaced by $s_{A}(x)$ and $b_{\tau}$ replaced by $b(x)+v(b,\kappa)
\tau$) we can assert that
\begin{equation}\label{DotZ1}
\dot{u}(x, t(x)+\tau) \geq  \uplambda-\kappa m (b(x) + v(b,\kappa)
\tau) e^{-s_{A}(x)(b(x)+v(b,\kappa)
\tau)} \geq v(b,\kappa),
\end{equation}
for all $\tau >0$ and
$t(x)=(b(x)-b)/v(b,\kappa)$.
This finishes the proof.
\end{proof}

\subsection{Expansion of aggregation: bounded support of $\phi$}

In this subsection we study how aggregation expands in the case when
$\phi$ has a bounded support. To do this we consider a slight modification of the method used in the previous subsection.
%It is clear that any function $\phi$
%from the previous section can be approximated by a function with a
%bounded support.
For simplicity of notations and further constructions we assume
that
\[
d=1, \qquad \phi(x)=\chi_{\bigl[-\frac{1}{2}, \frac{1}{2}\bigr]}(x), \ \ x\in\R, \qquad
A=[-a,a],\ \ a\ge1/2.
\]
We consider the function, cf. \eqref{dopeq},
\begin{equation}\label{rttr}
 r(b):=be^{-\frac{b}{4}},\quad b\in\R.
\end{equation}
It takes  maximal value $4/e$ at $b=4$. We define $\hat{b}$ to be the larger (or unique) solution to the equation $r(b)=\frac{\uplambda}{4m}$ provided $\uplambda\leq\frac{16m}{ e}$. If $\uplambda>\frac{16m}{e}$, we choose $\hat{b}\geq {4}$ to ensure
$be^{-\frac{b}{4}}=\frac{\uplambda}{4m}e^{-\frac{1}{4}\hat{b}}$
has two solutions, that is just a condition $\uplambda<\frac{16m}{ e}e^{\frac{1}{4}\hat{b}}$. It is worth noting that the function $r$ decays for $b>\hat{b}$ and the equation
\begin{equation}\label{sdasda32}
be^{-\frac{b}{4}}=\frac{\uplambda}{4m}e^{-\frac{b'}{4}}
\end{equation}
has always two solutions, for any $b'>\hat{b}$.

\begin{theorem}
Let $u_0\in C_b(\R)$ be an initial condition to \eqref{CP} such that
\eqref{u0m} holds
for some $b>\hat{b}$ and $\kappa=2$. Then the statement of Theorem~\ref{l-th1} holds true. Moreover, there exists $C>0$ such that
\begin{equation}
t(x) \le C |x| \ln |x| +o(|x|\ln |x|)\quad \mbox{ as } \quad |x| \to \infty.\label{txx}
\end{equation}
\end{theorem}
\begin{proof}
We begin by proving \eqref{txx} for all $x$ in a neighborhood of $A$. We define $B_1$ to be a segment of length $\frac{1}{2}$ with center at the right edge of $A$. Namely,
\begin{align*}
B_1&:=\Bigr[a-\frac{1}{4},a+\frac{1}{4}\Bigr].
\end{align*}
Set
\begin{align*}
S_1:= B_1 \cap A = \Bigr[a-\frac14, a\Bigl], \quad
U_1: =B_{1}\setminus S_{1}=\Bigl(a, a+ \frac14\Bigr].
\end{align*}
We next study the behavior of $u(x,t)$ at any point $ x \in U_1$. As before, \eqref{repres} and \eqref{estforder} yield
\begin{equation}\label{yepoyre}
 u(x, \tau) \leq 2 b + \uplambda \tau, \quad \tau >0, \quad x \in U_1.
\end{equation}
According to Theorem ~\ref{l-th1}, the estimate \eqref{Gr} implies
\begin{equation*}\label{ryeore}
u(y, \tau)> b + \frac{\uplambda }{2} \tau, \quad y \in S_1\subset A,
\end{equation*}
for any $\tau >0$. Thus
$$
(\phi \ast u)(x, \tau) \geq \int_{S_1} \phi(x-y)u(y,\tau)\,dy> \Bigl(b + \frac{\uplambda }{2} \tau\Bigr)\int_{S_1}\phi(x-y)\,dy.
$$
Since $x \in U_1$, $y\in S_1$, we have $x-y\in\bigl[0,\frac{1}{2}\bigr]$, and so
\begin{equation}\label{fasr3}
 (\phi \ast u)(x, \tau) >\frac{1}{4}\Bigl(b + \frac{\uplambda }{2} \tau\Bigr), \quad \tau>0, \quad x\in U_1.
\end{equation}
Combining \eqref{yepoyre} with \eqref{fasr3} and using the fact that the function $r$ decreases for $b > \hat{b}$, we get
\begin{align*}
m u(x,\tau) e^{- (\phi \ast u)(x, \tau)} &< 2 m \Bigl(b +
\frac{\uplambda }{2} \tau\Bigr) e^{-\frac{1}{4}  \bigl(b + \frac{\uplambda }{2}
\tau\bigr)} < 2 m \hat{b} e^{- \frac14  \hat{b}} \leq
\frac{\uplambda }{2},
\end{align*}
for any $\tau >0$ and any $x \in U_1$.  As a result,
\begin{equation}\label{DotZ}
\dot{u}(x, \tau) > \frac{\uplambda }{2}, \quad \tau \ge 0, \quad
x \in U_1.
\end{equation}
That means that $u(x,t)$ is monotonically
growing to infinity on $U_1$ with speed estimated by \eqref{DotZ}.

We can now proceed analogously by considering the next segment $B_2$ of the
length $\frac12$ with center at the right edge of $U_1$. Set
$$
U_2 := B_2 \setminus U_1, \quad t_1 := \frac{2}{\uplambda }(d_1 - d_0),
$$
where $d_0:=b$ and $d_1$ is, by definition, the larger solution to the equation, cf. \eqref{sdasda32},
\begin{equation}\label{d2}
r(b) = \frac{\uplambda }{4m} e^{-\frac{d_0}{4}}.
\end{equation}
We now show that the function $u(x,t)$ grows monotonically on $U_2$ for all $t \ge t_1$. Namely, by \eqref{DotZ}, we have on $U_1$
the following uniform estimate from below:
\begin{equation}\label{F1}
u(x, t_1) \ge u(x, 0) + \frac{\uplambda }{2} t_1 \ge \frac{\uplambda }{2}
t_1 = d_1 - d_0, \quad x \in U_1.
\end{equation}
Then,
\begin{equation}\label{dassaddsa}
 (\phi*u)(x,t_1)\geq (d_1-d_0)\int_{U_1}\phi(x-y)\,dy=\frac{1}{4}(d_1-d_0), \quad x\in U_2.
\end{equation}
On the other hand, by \eqref{repres} and \eqref{estforder}, in a new region $U_2$ the following estimate from
above holds:
\begin{equation}\label{F2}
u(x, t_1) \le u(x, 0) + \uplambda t_1 \le 2b + 2(d_1 - d_0) = 2 d_1,
\quad x \in U_2.
\end{equation}
%Therefore, if we define now $d_1$ as the maximal root of the equation, cf. \eqref{sdasda32},
%\begin{equation}\label{d2}
%r(b) = \frac{\uplambda }{4m} e^{-
%\frac{1}{4} d_0}, \quad b > d_0,
%\end{equation}
Combining (\ref{dassaddsa})--(\ref{d2}) we can assert that
\begin{equation}\label{d22}
m u(x, t_1) e^{- (\phi \ast u)(x, t_1 )} \le 2 m d_1
e^{-\frac{1}{4}  (d_1 - d_0)} = \frac{\uplambda }{2}, \quad x \in
U_2.
\end{equation}
Moreover, for any $t= t_1 + \tau, \; \tau \ge 0$ the inequality
(\ref{d22}) is also valid on $U_2$. Indeed, using \eqref{DotZ}, one can obtain in the same way as in (\ref{dassaddsa})--(\ref{d22}) the following
\begin{align*}
&m u(x, t_1+\tau) e^{- (\phi \ast u)(x, t_1 +\tau)} \le m (2
d_1 + \uplambda \tau) e^{-\frac{1}{4} (d_1 - b + \frac{\uplambda }{2}
\tau)} \\=&\,
 2 m (d_1 + \frac{\uplambda }{2} \tau) e^{- \frac{1}{4} (d_1 -b +
\frac{\uplambda }{2} \tau)} < 2 m d_1 e^{- \frac{d_1}{4}}
e^{\frac{b}{4} } = \frac{\uplambda }{2}, \quad x \in U_2,
\end{align*}
since $d_1+\frac{\uplambda }{2} \tau>d_1>d_0=b>\hat{b}$ and the function $r$ decreases for $b>\hat{b}$.
Therefore,
\begin{equation}\label{Dd2}
\dot{u}(x, t_1+\tau) > \frac{\uplambda }{2}, \quad \tau \ge 0,
\ \ x \in U_2,
\end{equation}
and, as a result, $u(x,t)$ is monotonically growing on $U_2$
for $t \ge t_1$.

This procedure can be continued step by step using segments $B_k, \;
k=1,2,\ldots$ of the same length $\frac12$ shifted to the right of $U_{k-1}$ with step $\frac14$. The same arguments may be applied in the left direction form the origin $0$ to get estimate \eqref{txx} for all negative $x$. Finally we have the
iterative relation for $t_k$ and $d_k$:
\begin{equation}\label{IT}
t_k = \frac{2}{\uplambda }(d_k - d_{k-1}), \quad d_k e^{-
\frac{d_k}{4}} = \frac{\uplambda }{4m} e^{- \frac{d_{k-1}}{4}}.
\end{equation}
%Since, for any $x\in \R$, $|x|\in B_{k(x)}$, where
Let $k(x)$ be the number of steps in our scheme required to reach $x\in \R$.
%It is evident that $k(x)\leq [4|x|]+1$ for large enough $|x|$,
Then the time $t(x)$ at which the
function $u(x,t)$ starts to increase monotonically at $x$ with
$$
\dot{u}(x,t) > \frac{\uplambda }{2}, \quad t \ge t(x),
$$
can be bounded from above as
$$
t(x) \leq t_1+t_2+\ldots+t_{k(x)}= \frac{2}{\uplambda }(d_{k(x)}-b).
$$
In particular,
$$
t(x) \leq \frac{2}{\uplambda } d_{k(x)},
$$
where $d_{k(x)}$ is defined by recurrence relation (\ref{IT}). To
complete the proof of Theorem 5.7 we use the asymptotic for $d_k$.

We set
\[
c_k:=\frac{d_k}{4}, \qquad \mu:=\ln \frac{\uplambda}{16m},
\]
then \eqref{IT} yields
\begin{equation}\label{ITmod}
 c_ke^{-c_k}= e^\mu e^{-c_{k-1}}.
\end{equation}
We recall that we choose $c_0={b}/{4}$ in such a way that, for $k=1$, \eqref{ITmod} has a solution $c_1>c_0$, hence, for $k=2$, \eqref{ITmod} has a solution $c_2>c_1$ and so on. Taking logarithms of both sides of \eqref{ITmod}, we get
\begin{equation}\label{ITmod2}
 c_k-\ln c_k+\mu=c_{k-1}.
\end{equation}
The statement of the theorem will be proved once we prove the proposition below.
\begin{proposition}
The sequence of $c_k$ defined by a
recurrence relation \eqref{ITmod2} has the following asymptotic
representation
\begin{equation}\label{DK}
c_k = k \ln k + k \ln \ln k -(\mu+1) k + o(1), \quad k \to \infty.
\end{equation}
\end{proposition}
\begin{proof}
Substituting $c_k=k \ln k + k \ln \ln k + a k+g_k$ with $g_k=o(1)$, $k \to \infty$ into \eqref{ITmod2} we deduce that
 \begin{align*}
 \ln c_k & =\ln \bigl( k \ln k + k \ln \ln k + a k+g_k\bigr) \\
 & = \ln \bigl( k \ln k \bigr) +\ln\Bigl( 1+\frac{k \ln \ln k + a k+g_k}{k\ln k}\Bigr)\\
 & = \ln k + \ln\ln k +\frac{k \ln \ln k + a k+g_k}{k\ln k}+o(1), \quad k\rightarrow\infty
 \end{align*}
 and
 \begin{align*}
 c_{k-1}&=(k-1) \ln \Bigl(k\Bigl(1-\frac{1}{k}\Bigr)\Bigr) + (k -1)\ln \ln \Bigl(k\Bigl(1-\frac{1}{k}\Bigr)\Bigr) + a (k-1)+g_{k-1}\\
 &=(k-1)\ln k -(k-1)\frac{1}{k}+(k-1)\ln\Bigl(\ln k \Bigl(1-\frac{1}{k\ln k}\Bigr)\Bigr)+a(k-1)+o(1)\\
 &=(k-1)\ln k -(k-1)\frac{1}{k}+(k-1)\ln\ln k -(k-1)\frac{1}{k\ln k}+a(k-1)+o(1).
 \end{align*}
 As a result,
 \begin{align*}
 &\quad c_k-\ln c_k+\mu-c_{k-1}\\&=k \ln k + k \ln \ln k + a k -\ln k - \ln\ln k -\frac{k \ln \ln k + a k+g_k}{k\ln k}\\
 &\quad - (k-1)\ln k +(k-1)\frac{1}{k}-(k-1)\ln\ln k +(k-1)\frac{1}{k\ln k}-a(k-1)+\mu+o(1)\\
 &=o(1),
 \end{align*}
 provided that $a=-\mu-1$.
\end{proof}
The theorem is fully proved.
\end{proof}

\begin{remark}
 The relation \eqref{ITmod} may be rewritten using the so-called tree function $T(x)=-W(-x)$, where $W$ is the Lambert $W$-function and $x\in\bigl(-\frac{1}{e},0\bigl)$, see e.g. \cite{BFS2008} and the references therein. More precisely, the Lambert $W$-function which solves equation $W(x)e^{W(x)}=x$, $x\geq-\frac{1}{e}$ has two real branches: the principal branch $W_0$ which increases on $\bigl(-\frac{1}{e},+\infty\bigr)$ and the negative branch $W_{-1}$ which decreases on $\bigl(-\frac{1}{e},0\bigr)$. We are interested in the latter branch, namely, \eqref{ITmod} may be rewritten as
 \begin{equation}\label{IT_LW}
 c_k=-W_{-1}\bigl(-e^{\mu-c_{k-1}}\bigr), \quad c_{0}>\mu+1.
\end{equation}
Up to our knowledge, the asymptotic \eqref{DK} is new for the iteration \eqref{IT_LW} of the Lambert $W$-function.
\end{remark}

%GATHER{D:/_TeX/MyArticles/actual.bib}
%\bibliographystyle{is-abbrv}
%\bibliography{Macintosh HD/Users/Dr. O.Kutovyi /Documents/Math/2013/Articles/with Lena Zhizhina}
\def\cprime{$'$}

\end{document}